\documentclass[12pt]{article}  

\usepackage{graphicx}
\usepackage{amsmath}
\usepackage{amsfonts}
\usepackage{amssymb}
\usepackage{amsthm}
\usepackage{color}
\usepackage{centernot}
\usepackage{mathtools}
\usepackage{stmaryrd}

\newtheorem{theorem}{Theorem}
\newtheorem{proposition}{Proposition}
\newtheorem{corollary}{Corollary}
\newtheorem{lemma}{Lemma}

\theoremstyle{definition}
\newtheorem{definition}{Definition}

\begin{document}

\title{Spontaneous stochasticity and renormalization group in discrete multi-scale dynamics}
\date{\today}
\author{Alexei A. Mailybaev\footnote{Instituto de Matem\'atica Pura e Aplicada -- IMPA, Rio de Janeiro, Brazil. E-mail: alexei@impa.br} \and Artem Raibekas\footnote{Instituto de Matem\'atica e Estat\'istica, UFF, Niter\'oi, Brazil. E-mail: artemr@id.uff.br}}

\maketitle

\begin{abstract}
We introduce a class of multi-scale systems with discrete time, motivated by the problem of inviscid limit in fluid dynamics in the presence of small-scale noise. These systems are infinite-dimensional and defined on a scale-invariant space-time lattice.
We propose a qualitative theory describing the vanishing regularization (inviscid) limit as an attractor of the renormalization group operator acting in the space of flow maps or respective probability kernels. If the attractor is a nontrivial probability kernel, we say that the inviscid limit is spontaneously stochastic: it defines a stochastic (Markov) process solving deterministic equations with deterministic initial and boundary conditions. The results are illustrated with solvable models: symbolic systems leading to digital turbulence and systems of expanding interacting phases.
\end{abstract}

\section{Introduction}

Space-time scale invariance is a fundamental property of many physical models. 
In ideal fluid dynamics, it refers to transformations of the velocity field $\mathbf{u}(\mathbf{r},t)$ as~\cite{frisch1999turbulence}
	\begin{equation}
	\label{eq_inv}
	t,\mathbf{r},\mathbf{u} \mapsto \lambda^{1-h}t,\lambda\mathbf{r},\lambda^h \mathbf{u},
	\end{equation}
where $\lambda > 0$ is a scaling factor and $h \in \mathbb{R}$ is a space-time scaling exponent. A textbook example is the inviscid Burgers equation as a prototype for compressible ideal fluid, while our main motivation comes from the developed turbulence, where symmetries (\ref{eq_inv}) refer to Euler equations for incompressible ideal fluid. 
A common feature of ideal scale-invariant models is that their solutions may be non-unique or not globally defined~\cite{eggers2008role,buckmaster2021convex}. A globally well-posed system is obtained by adding regularizing (e.g., viscous) terms, and then a physically relevant solution for the ideal system is selected in the inviscid limit. 
The important aspect, to which we put special attention in this paper, is the effect of small-scale fluctuations. Earlier studies~\cite{lorenz1969predictability,ruelle1979microscopic,leith1972predictability,eyink1996turbulence,boffetta2001predictability,falkovich2001particles,biferale2018rayleigh} suggest that turbulent dynamics is intrinsically stochastic even when the noise is tiny and limited to small scales, and that the stochastic behavior persists when this noise is removed in the limit of vanishing regularization~\cite{mailybaev2016spontaneously,thalabard2020butterfly}; see also \cite{drivas2020statistical,eyink2020renormalization,mailybaev2021spontaneously} for solvable mathematical examples. Such limiting solutions are called \textit{spontaneously stochastic}: they represent probability distributions on a set of non-unique solutions of deterministic equations with deterministic initial conditions. 

Our present work aims at developing a qualitative theory of the inviscid limit in a class discrete-time scale-invariant models. This theory explains why the inviscid limit can be deterministic or spontaneous stochastic and why it can be universal, i.e., not sensitive to the choice of regularizing terms. According to the symmetry (\ref{eq_inv}), we consider a geometric sequence of spatial scales $\ell_n = \lambda^{-n}$ and corresponding temporal scales (turn-over times) $\tau_n = \ell_n^{1-h}$. At each scale $\ell_n$, the system is characterized by a variable $u_n(t)$. For example, such kind of scales and variables in fluid dynamics result from the Littlewood--Paley decomposition of the  velocity field. For our study we select a specific symmetry with $\lambda = 2$ and $h = 0$. This symmetry facilitates the discrete-time formulation with the spatial and temporal scales 
	\begin{equation}
	\label{eq1}
	\ell_n = \tau_n = 2^{-n}, \quad n = 0,1,2,\ldots.
	\end{equation}
Regularization is introduced by modifying the dynamics below a small ``viscous scale'' $\ell_N$. This regularization is removed in the inviscid limit $\ell_N \to 0$ as $N \to \infty$. 

Our main result is that the inviscid limit is governed by the dynamical system 
	\begin{equation}
	\label{Int_1}
	\psi^{(N+1)} = \mathcal{R}_g[\psi^{(N)}], 
	\end{equation}
where $\psi^{(N)}$ is a turn-over time evolution map for the system regularized at scale $\ell_N$, and 
	\begin{equation}
	\label{Int_1R}
	\mathcal{R}_g[\psi] = \sigma_- \circ \psi \circ \psi \circ \sigma_+ +\xi
	\end{equation}
is the re\-nor\-ma\-li\-za\-ti\-on-group (RG) operator. Here, the last term $\xi$ depends on the ideal model, while the first term $\sigma_- \circ \psi \circ \psi \circ \sigma_+$ with scaling maps $\sigma_\pm$  is similar to the Feigenbaum--Cvitanovi\'c functional relation~\cite{feigenbaum1983universal} (introduced independently by Coullet--Tresser~\cite{coullet1978iterations}). 
The initial map $\psi^{(1)}$ depends on the choice of regularization. In the presence of noise, which is modelled by random fluctuations at regularized scales, we introduce a stochastic version of the RG operator acting in a space of probability (Markov) kernels.

Our theory relates the inviscid limit $N \to \infty$ with attractors of the RG operator. A fixed point attractor, $\mathcal{R}_g^N[\psi^{(1)}] \to \psi^\infty$, yields a unique and universal deterministic dynamics of the ideal system. Analogous limit for the Burgers equation is known as shock solutions; see e.g. \cite{dafermos2005hyperbolic}. When the attractor is a nontrivial probability kernel, the limiting dynamics is spontaneously stochastic: it represents a stochastic (Markov) process solving deterministic equations of ideal system with deterministic initial and boundary conditions.

We present several examples of relatively simple systems, for which our results are verified analytically and numerically: symbolic models and systems with interacting phases. In symbolic systems, the variables ${u}_n(t) \in \mathbb{Z}^2 = \{0,1\}$ take two values interpreted as laminar and turbulent states. We argue that, like in the theory of dynamical systems \cite[\S1.9]{katok1997introduction}, ``in many respects symbolic systems serve as models for smooth ones; it is often easier to see many properties in the symbolic case first and then carry them over to the smooth case''. Indeed, these models demonstrate such properties as a finite-time blowup and, depending on the form of interactions, a deterministic or spontaneously stochastic behavior in the inviscid limit. We refer to their irregular multi-scale dynamics as \textit{digital turbulence}. The second example belongs to a class of solvable spontaneously stochastic models with linear expanding phase interactions. Such models, where the interactions are given by hyperbolic toral automorphisms, e.g. Arnold's cat map, were studied in~\cite{mailybaev2021spontaneously}.

The paper is organized as follows. In Section \ref{sec2} we introduce a general model on a scale-invariant lattice. 
Afterwards, regularized solutions and the corresponding RG operator are defined in Section \ref{sec_DR}. In Section~\ref{section4} we associate solutions obtained in the inviscid limit with the fixed-point attractor of the RG operator. The class of symbolic models is studied in Section~\ref{sec_ex} demonstrating different types of convergence for the RG operator. Motivated by the  lack of convergence in the previous example, we extend our results in Section~\ref{sec8} to the stochastic form of regularization. Section~\ref{sec9} associates spontaneously stochastic solutions with an attractor of the stochastic RG operator. We then present in Section~\ref{sec_sp_ex} two examples of spontaneously stochastic systems: models with expanding interacting phases and symbolic models. Finally, Section~\ref{sec10} summarizes the obtained results and discusses directions for further research.
In the Appendix, we collect some rigorous results for symbolic models and technical derivations concerning non-integer times. 

\section{Model}\label{sec2}

We start with an informal description, which motivates a class of discrete-time systems studied in the paper. 
Let a system be represented by a geometric sequence of scales (\ref{eq1}) and described at every scale and time
by a variable ${u}_n(t)$ belonging to some phase space $X$. 
We consider scaling symmetries (\ref{eq_inv}) with $\lambda = 2^m$ and $h = 0$ written for variables ${u}_n(t)$ as
	\begin{equation}
	\label{eq5}
 	t,{u}_n \mapsto 2^m t, {u}_{n+m}, \quad m \in \mathbb{Z},
 	\end{equation}
where the change of indices reflects the change of scales $\ell_n = 2^m \ell_{n+m}$.
Then, a general scale-invariant equation of motion with local inter-scale interactions can be formulated as
	\begin{equation}
	\label{eqSI1}
	\frac{d{u}_n}{dt} = \frac{F({u}_{n-i},\ldots,{u}_n,\ldots,{u}_{n+i})}{\tau_n},
	\end{equation}
where the function $F$ describes interactions with $i$ nearby scales. 
Scale invariant equations with a discrete time follow, e.g., by using a finite-difference approximation of equation (\ref{eqSI1}) with a time step proportional to $\tau_n$. The simplest version is given by the Euler method with $\Delta t = \tau_n$ as
	\begin{equation}
	\label{eqSI3}
	{u}_n(t) = {u}'_n+F({u}'_{n-i},\ldots,{u}'_n,\ldots,{u}'_{n+i}),
	\end{equation}
where the primes denote the variables evaluated at time $t-\tau_n$. 

\begin{figure}[t]
\centering
\includegraphics[width=0.8\textwidth]{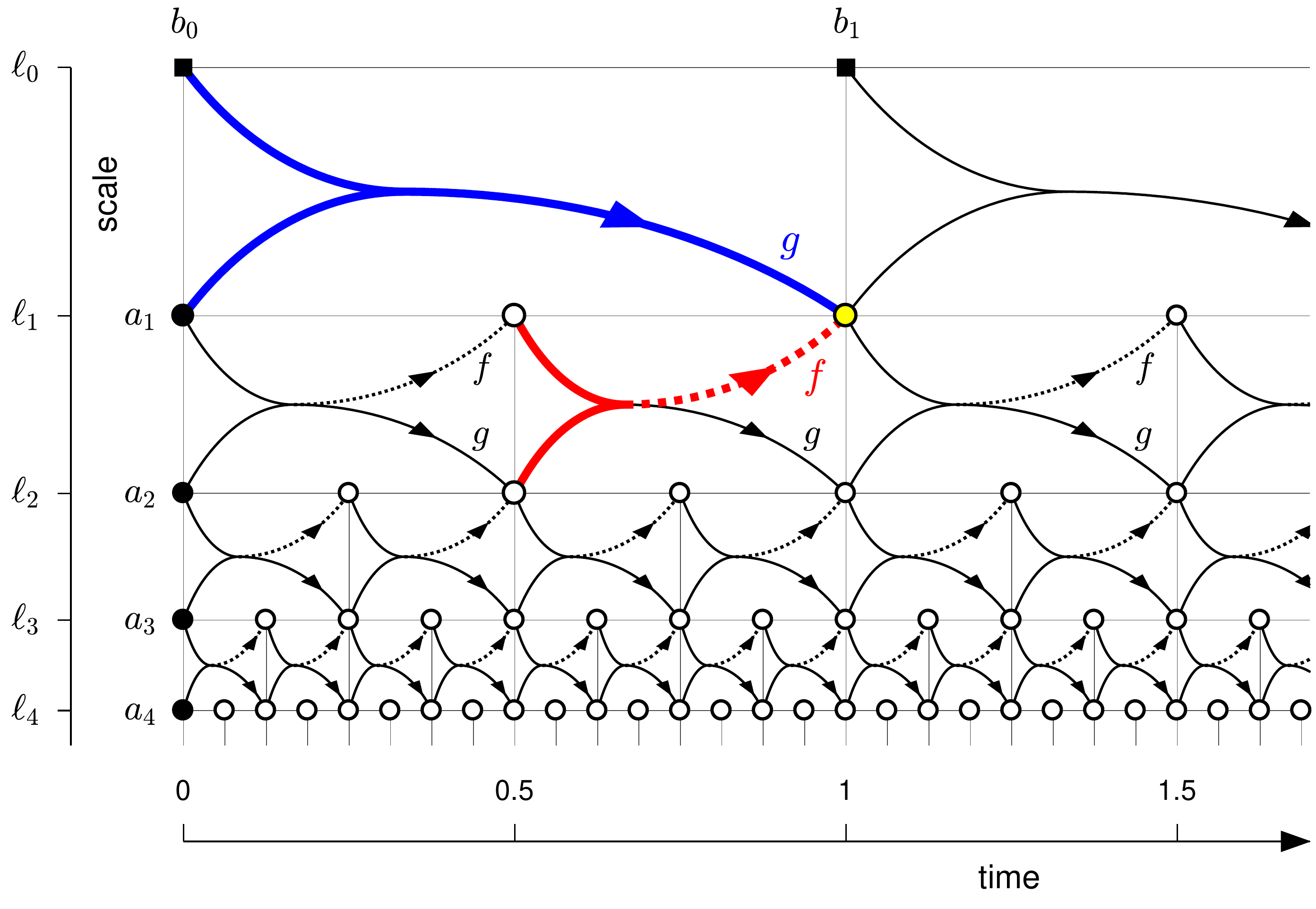}
\caption{Structure of the multi-scale space-time lattice $\mathcal{L}$. Dynamic variables ${u}_n(t)$ at scales $\ell_n$ and times $t > 0$ correspond to empty circles. Full squares correspond to boundary conditions $b_t$ at the scale $\ell_0$ and full circles to initial conditions $a_n$ at $t = 0$. Arrows indicate inter-scale interactions: interaction with a smaller scale (dotted arrow) is governed by the function $f$, and interaction with a larger scale (solid arrow) is governed by the function $g$. The bold blue and red lines illustrate the governing relation (\ref{eq4}) at $(n,t) = (1,1)$.}
\label{fig1}
\end{figure}

We now use equation (\ref{eqSI3}) as a motivation and introduce a convenient class of ideal scale-invariant models. With different time steps $\tau_n$ at different scales, the discrete space-time becomes the multi-scale lattice 
	\begin{equation}
	\label{eq2}
	\mathcal{L} = \{(n,t): t = m\tau_n,\ n,m\in\mathbb{N}\},
	\end{equation}
shown in Fig.~\ref{fig1}. At each point of the lattice $(n,t) \in \mathcal{L}$ we consider a variable ${u}_n(t) \in X$, where $X$ is a complete separable metric space with an additive group operation.  We restrict inter-scale interactions to nearest neightbors with the purpose of facilitating the further analysis. Specifically, we consider interactions described by continuous functions $f: X^2 \mapsto X$ and $g: X^2 \mapsto X$ within each cell as shown in Fig.~\ref{fig1}. Thus, the governing equations of our system are formulated as
	\begin{equation}
	\label{eq4}
	{u}_n(t) = \left\{ 
	\begin{array}{ll}
	f\big({u}'_n,{u}'_{n+1} \big),
	& t/\tau_n \textrm{ is odd};
	\\[5pt]
	f\big({u}'_n,{u}'_{n+1} \big)+g\big({u}''_{n-1},{u}''_{n} \big),
	& t/\tau_n \textrm{ is even};
	\end{array}
	\right.
	\end{equation}
where a single primes denotes the value at time $t-\tau_n$ and a double prime at $t-2\tau_{n}$.
These equations have the scaling symmetry (\ref{eq5}) and express any variable ${u}_n(t)$ in terms of variables at earlier times and adjacent scales.

As shown in Fig.~\ref{fig1}, we introduce the initial conditions at time $t = 0$ as
	\begin{equation}
	\label{eq4c}
	{u}_n(0) = a_n, \quad n \in \mathbb{N} = \{1,2,3,\ldots\},
	\end{equation}
and the large-scale boundary (forcing) conditions associated with $\ell_0 = 1$ as
	\begin{equation}
	\label{eq4b}
	{u}_0(t) = b_t, \quad t \in \mathbb{N}_0 = \{0,1,2,\ldots\}.
	\end{equation}

Our goal is to study the initial value problem given by equations (\ref{eq4}) at all points of the lattice $\mathcal{L}$ with initial conditions (\ref{eq4c}) and boundary conditions (\ref{eq4b}). We refer to this problem as the ideal system, unlike a regularized system introduced later. By analogy with partial differential equations, we refer to solutions of the ideal system as \textit{weak solutions}, since no additional condition on their regularity (limitations on variables at small scales) is imposed. Existence and uniqueness of weak solutions is a nontrivial problem, as shown in the example below. 

We remark that the specific form of equation (\ref{eq4}) is chosen for convenience, facilitating the formulation of renormalization group theory. 
Our formulation is flexible in terms of a configuration space $X$, including both simple symbolic models considered below and physically relevant models obtained with a scale-by-scale separation of degrees of freedom~\cite{frisch1999turbulence}; see also the discussion in Section~\ref{sec10}.

\subsection{Symbolic models} \label{sec_symb_ex}
Let us consider symbolic models with $X =  \mathbb{Z}_2 = \{0,1\}$. We introduce two specific models, which are used later for demonstrating deterministic and spontaneously stochastic behaviors. One can interpret the values ${u}_n(t) = 0$ and $1$, respectively, as a ``laminar'' and ``turbulent'' state at scale $\ell_n$ and time $t$. The first model is defined by the functions 
	\begin{equation}
	\label{eq5fg}
	\textrm{A}:\ f(0,0) = f(0,1) = f(1,1) = g(0,0) = 0,\quad f(1,0) = g(0,1) = g(1,0) = g(1,1) = 1.
 	\end{equation}
Figure~\ref{fig2}(a) shows a solution at times $t \le 1$ for specific initial and boundary conditions. Non-zero components of this solution are restricted to large scales at $t = 0$, but propagate to infinitely small scales $\ell_n \to 0$ (large $n$) as $t \to 1$. The second model corresponds to 
	\begin{equation}
	\label{eq13}
	\textrm{B}:\ f(0,0) = f(0,1) = f(1,0) = g(0,0) = g(1,1) = 0, \quad
	f(1,1) = g(0,1) = g(1,0) = 1.
 	\end{equation}
Figure~\ref{fig2}(b) shows a solution for specific initial and boundary conditions with a similar behavior: non-zero components reach infinitely small scales as $t \to 0.5$.

\begin{figure}[t]
\centering
\includegraphics[width=1\textwidth]{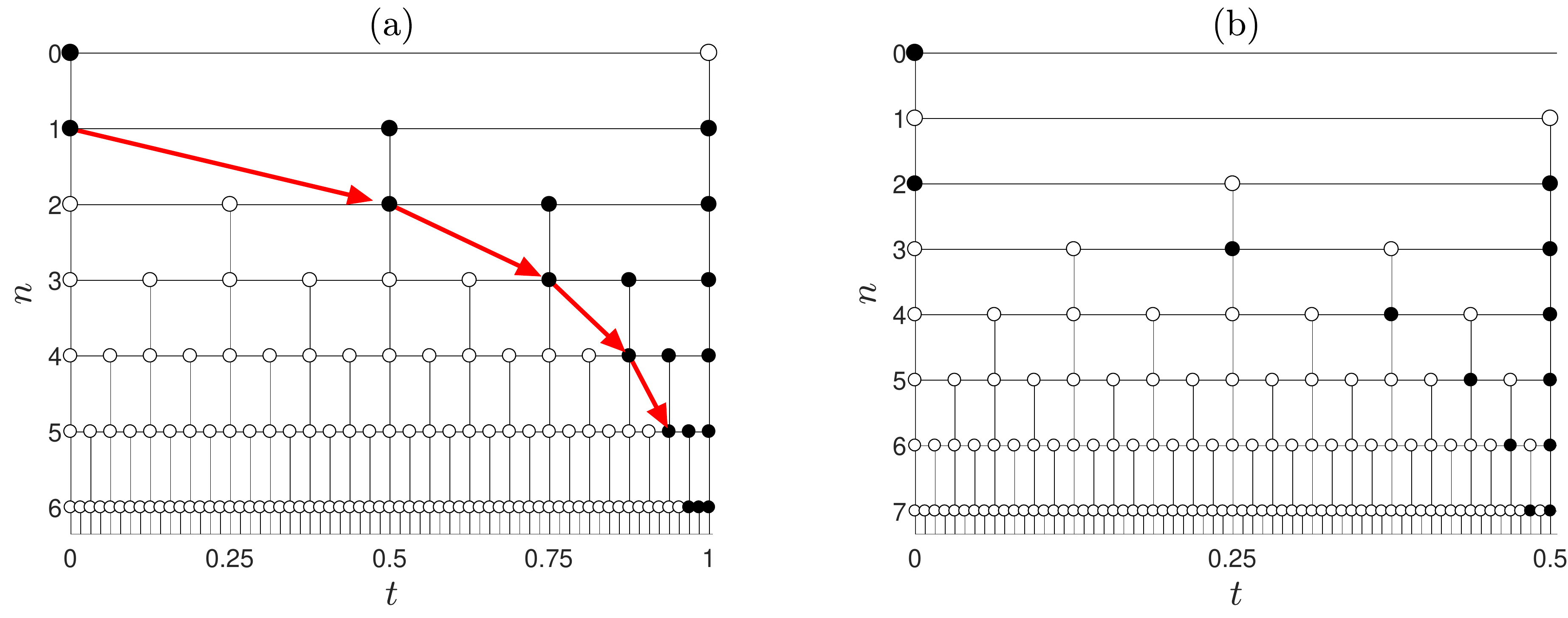}
\caption{Solutions of symbolic multi-scale models with a finite-time blowup. (a) Model A with initial conditions $(a_1,a_2,\ldots) = (1,0,0,\ldots)$ and boundary conditions $(b_0,b_1) = (1,0)$. (b) Model B with initial conditions $(0,1,0,0,\ldots)$ and boundary condition $b_0 = 1$. Black circles correspond to ${u}_n(t) = 1$ and empty circles to ${u}_n(t) = 0$. The arrows indicate the development of blowup for the strong solution in figure (a).}
\label{fig2}
\end{figure}

The presented solutions have an interesting interpretation in the context of strong and weak solutions of partial differential equations, in particular, the Euler equations in fluid dynamics \cite{buckmaster2021convex}. We elaborate this interpretation in more detail in Appendix \ref{sec3} and summarize in the rest of this section. Solutions ${u}_n(t)$ of the ideal system (\ref{eq4})--(\ref{eq4b}) may be regarded as \textit{weak solutions}, since we pose no limitation on nonzero components ${u}_n(t) = 1$ at small scales $\ell_n \to 0$. This is essential because small scales evolve with turnover times $\tau_n = \ell_n \to 0$, i.e., very fast.  As a consequence, weak solutions exist globally in time but they are generally nonunique; see Proposition~\ref{th2} in the Appendix. We can define \textit{strong solutions} by an extra condition that the variables ${u}_n(t)$ vanish for sufficiently small scales. Given initial and boundary conditions, the strong solution exists and is unique locally in time; see Definition~\ref{def1} and Proposition~\ref{th1}. 

However, for strong solutions, nonzero components can propagate to arbitrarily small scales in finite time, as $t \to T$. We call this situation a finite-time singularity or \textit{blowup}. This is shown in Fig.~\ref{fig2} for both models, where the blowup time is (a) $T = 1$ and (b) $T = 0.5$. The blowup demonstrated in these very simple models follows a self-similar scenario, which is also typical for partial differential equations and turbulence models; see e.g.~\cite{dombre1998intermittency,eggers2008role,mailybaev2012renormalization}. Anticipating the phenomenon of spontaneous stochasticity, we prove another interesting statement in Proposition~\ref{prop2}: weak solutions in the model B are unique until the blowup time, but non-unique at larger times $t > T$.

\section{Regularized solutions} \label{sec_DR}

Since solutions of the ideal system (\ref{eq4})--(\ref{eq4b}) are generally nonunique, we now introduce a physically motivated selection procedure based on regularization. We start with the simplest form of the regularization: we choose a scale number $N \in \mathbb{N}$ and assume that there is no evolution at scales smaller than $\ell_N$.

\begin{definition} \label{def_det}
Given a regularization scale number $N$, we introduce the \textit{regularized solution} ${u}_n^{(N)}(t)$ denoted by the superscript $(N)$. It is defined by setting
	\begin{equation}
	\label{eq14}
	{u}_n^{(N)}(t) = 0, \quad n > N, \quad t \ge 0,
 	\end{equation}
and determining the remaining components ${u}_n^{(N)}(t)$, $n \le N$, by equations (\ref{eq4}) with initial conditions (\ref{eq4c}) for $n = 1,\ldots,N$ and boundary conditions (\ref{eq4b}). 
\end{definition}

Let us introduce the space of sequences 
	\begin{equation}
	\label{eq15}
	{X^{\mathbb{N}}} = \left\{{{u}} = ({u}_1,{u}_2,\ldots): {u}_n \in X, n \in \mathbb{N}\right\},
 	\end{equation}
considered as the infinite product space with the product topology. It is a complete separable metric space; see e.g. \cite{neveu1965mathematical}. We denote 
	\begin{equation}
	\label{eq_nt_a}
	a = (a_1,a_2,\ldots) \in X^{\mathbb{N}},\quad 
	{{u}}^{(N)}(t) = ({u}_1^{(N)}(t),{u}_2^{(N)}(t),\ldots) \in X^{\mathbb{N}},\quad
	t/\tau_1 \in \mathbb{N}_0,
 	\end{equation}
for integer and half-integer times $t$. 

It is easy to see that the regularized solutions are unique for any $N$: each variable ${u}_n^{(N)}(t)$ is determined by equations (\ref{eq4}) as a function of initial and boundary conditions after a finite number of iterations. 
Observe that any subsequence $N_i \to \infty$, such that the variables $u_n^{(N_i)}(t)$ converge at all points of the lattice $(n,t) \in \mathcal{L}$, defines a solution of the ideal system. However, this limit may not be  unique.
Notice the regularization scale $\ell_N$ is analogous to the Kolmogorov dissipative scale in fluid dynamics, where the flow is suppressed by viscosity at scales smaller than $\ell_N$~\cite{frisch1999turbulence}. By this analogy, we call $N \to \infty$ the \textit{inviscid limit}. 

We denote by $C({X^{\mathbb{N}}})$ a space of continuous maps from ${X^{\mathbb{N}}}$ to itself. 
The regularized solution at time $t = 1$ depends on initial conditions $a$ and the boundary condition $b_0$;  see Fig.~\ref{fig1}. Let us introduce the unit-time flow map in $C(X^{\mathbb{N}})$ as
	\begin{equation}
	\label{eq19flowAt0}
	\phi^{(N)}_{b_0} : a \mapsto {{u}}^{(N)}(1).
 	\end{equation}
Since our system is invariant with respect to the unit-time translation, one finds solutions at any integer time as
	\begin{equation}
	\label{eq19flowAtB}
	\phi^{(N)}_{b_{t-1}} \circ \cdots \circ \phi^{(N)}_{b_0}: 
	a \mapsto {u}^{(N)}(t), \quad t \in \mathbb{N}.
 	\end{equation}
Let us denote by $\psi^{(N)} \in C({X^{\mathbb{N}}})$ the flow map at half-time $\tau_1 = 1/2$:
	\begin{equation}
	\label{eq19flowA}
	\psi^{(N)} : a \mapsto {{u}}^{(N)}(\tau_1).
 	\end{equation}
This map does not depend on boundary conditions as one can see from Fig.~\ref{fig1}. If one neglects the interaction between scales $\ell_0$ and $\ell_1$, the unit-time flow map would be the composition $\psi^{(N)} \circ \psi^{(N)}$ as shown in Fig.~\ref{fig4}. One can see using equation (\ref{eq4}) that the omitted interaction affects only the variable ${u}^{(N)}_1(1)$ by an additive term $g(b_0,a_1)$. Hence, we express
	\begin{equation}
	\label{eq19flowAt}
	\phi^{(N)}_{b_0} = \psi^{(N)} \circ \psi^{(N)}+\beta_{b_0}
 	\end{equation}
with the map
	\begin{equation}
	\label{eq19flowX}
	\beta_{b_0}: a \mapsto \left(g(b_0,a_1),0,0,\ldots\right).
 	\end{equation}

\begin{figure}[t]
\centering
\includegraphics[width=0.63\textwidth]{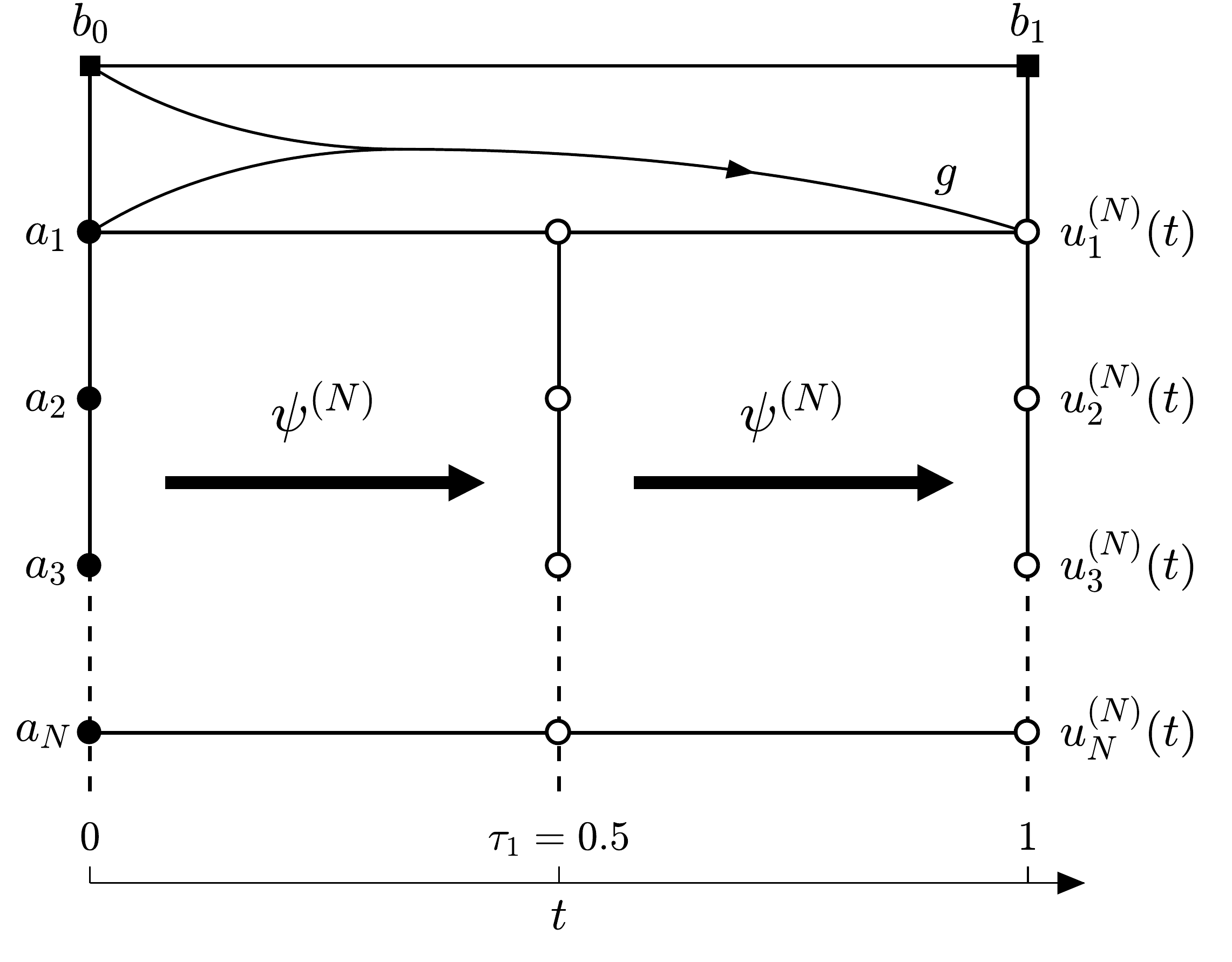}
\caption{Structure of the regularized dynamics determining the unit-time map $\phi_{b_0}^{(N)}: a \mapsto u^{(N)}(1)$ in terms of the half-time map $\psi^{(N)}: a \mapsto u^{(N)}(\tau_1)$.}
\label{fig4}
\end{figure}

We derive below the iteration relations for $\psi^{(N)}$ formulated using the shift (scaling) maps
	\begin{equation}
	\label{eq30Rs}
	\sigma_+: a \mapsto (a_2,a_3,\ldots), \quad
	\sigma_-: a \mapsto (0,a_1,a_2,\ldots),
 	\end{equation}
which increase or decrease all scales of a given state, and the coupling map
	\begin{equation}
	\label{eq30C}
	\xi: a \mapsto \left(f(a_1,a_2),\, g(a_1,a_2),\,0,0,\ldots\right).
 	\end{equation}

\begin{proposition}
\label{prop3} For any $N \in \mathbb{N}$, we have 
	\begin{equation}
	\label{eq30P}
	\psi^{(N+1)} = \sigma_- \circ \psi^{(N)} \circ \psi^{(N)} \circ \sigma_+ +\xi.
 	\end{equation}
\end{proposition}

\begin{proof} 
Consider a regularized solution ${u}_n^{(N+1)}(t)$. Following scaling symmetry (\ref{eq5}) with $m = 1$, we introduce new variables by changing both time and scale with the factor of 2 as  
	\begin{equation}
	\label{eq5N}
	\tilde{{u}}^{(N)}_n(t) = {u}_{n+1}^{(N+1)}(t/2).
 	\end{equation}
Notice that we also changed the regularization scale.
One can verify that variables (\ref{eq5N}) satisfy the initial conditions and the first boundary condition in the form
	\begin{equation}
	\label{eq5Nic_b}
	\tilde{{u}}^{(N)}(0) = \tilde{a} = \sigma_+(a),\quad
	\tilde{{u}}^{(N)}_0(0) = \tilde{b}_0 = a_1, 
 	\end{equation}
as well as regularization relations (\ref{eq14}) and equations of ideal system (\ref{eq4}). Thus, relation (\ref{eq5N}) describes the scaling symmetry 
	\begin{equation}
	\label{eq5reg}
 	t,{u}_n^{(N)} \mapsto 2t, {u}_{n+1}^{(N+1)}
 	\end{equation}
for the regularized system. Relations (\ref{eq19flowAt0}) and (\ref{eq19flowAt}) for solution (\ref{eq5N}) yield
	\begin{equation}
	\label{eq19flowPrA0}
	\phi^{({N})}_{\tilde{b}_0} : \tilde{a} \mapsto \tilde{{u}}^{({N})}(1),
	\quad
	\phi^{({N})}_{\tilde{b}_0} = \psi^{({N})} \circ \psi^{({N})}+\beta_{\tilde{b}_0}.
 	\end{equation}
In terms of original variables, these expressions take the form
	\begin{equation}
	\label{eq19flowPrA}
	\phi^{(N)}_{a_1} : (a_2,a_3,\ldots) \mapsto \left({u}_2^{(N+1)}(\tau_1),{u}_3^{(N+1)}(\tau_1),\ldots\right),
	\quad
	\phi^{(N)}_{a_1} = \psi^{(N)} \circ \psi^{(N)}+\beta_{a_1}.
 	\end{equation}
Expressing $\psi^{(N)} \circ \psi^{(N)}$ from (\ref{eq19flowPrA}) in the right-hand side of relation (\ref{eq30P}), we write
	\begin{equation}
	\label{eq19flowPrR}
	\sigma_- \circ \psi^{(N)} \circ \psi^{(N)} \circ \sigma_+ +\xi
	= \sigma_- \circ \left( \phi^{(N)}_{a_1}-\beta_{a_1}\right) \circ \sigma_+ +\xi.
 	\end{equation}
Using the explicit form of maps (\ref{eq19flowX})--(\ref{eq30C}), and the first relation in (\ref{eq19flowPrA}) one can see that the map (\ref{eq19flowPrR}) acts as
	\begin{equation}
	\label{eq19flowPrS}
	a \mapsto \left(f(a_1,a_2),{u}_2^{(N+1)}(\tau_1),{u}_3^{(N+1)}(\tau_1),\ldots\right).
 	\end{equation}
Since ${u}_1^{(N+1)}(\tau_1) = f(a_1,a_2)$ is given by equations (\ref{eq4}) and (\ref{eq4c}), and recalling definition (\ref{eq19flowA}), we find that relation (\ref{eq19flowPrS}) represents the map $\psi^{(N+1)}$.
\end{proof}

\section{Inviscid limit as RG dynamics}\label{section4}

Now let us describe the inviscid limit in terms of the \textit{renormalization group (RG) operator}.

\begin{definition}
The RG operator is the map from $C({X^{\mathbb{N}}})$ to itself defined as
	\begin{equation}
	\label{eq30R}
	\mathcal{R}_g: \psi \mapsto \sigma_- \circ \psi \circ \psi \circ \sigma_+ +\xi,\quad \psi \in C({X^{\mathbb{N}}}).
 	\end{equation}
\end{definition}

Proposition~\ref{prop3} yields the RG equation
	\begin{equation}
	\label{eq30}
	\psi^{(N+1)} = \mathcal{R}_g\left[\psi^{(N)}\right],
 	\end{equation}
which defines iteratively all regularized maps $\psi^{(N)}$ given the initial map for $N = 1$. The latter is found explicitly as
	\begin{equation}
	\label{eq19b}
	\psi^{(1)}: (a_1,a_2,\ldots) \mapsto \left(f(a_1,0),0,0,\ldots\right).
 	\end{equation}
We conclude that the inviscid limit $N \to \infty$ is governed by the dynamical system (\ref{eq30}) defined by the RG operator (\ref{eq30R}), in which the regularization scale number $N$ plays the role of  ``time''. The RG operator depends only on the ideal system via the coupling map $\xi$. One can recognise analogy of our RG operator (\ref{eq30R}) with the Feigenbaum--Cvitanovi\'c functional equation in the theory of dynamical systems~\cite{feigenbaum1983universal}: the term $\sigma_- \circ \psi \circ \psi \circ \sigma_+$ represents the iterated map rescaled with the factor $1/2$, similar to the composition $\alpha g \circ g(x/\alpha)$ with $\alpha \approx 2.5029$. 

In order to understand the dynamics of the RG operator, we use the following definition of convergence for maps in $C({X^{\mathbb{N}}})$.

\begin{definition}\label{def_lim}
We say that $\psi^{(N)} \to \psi^\infty$ is a limiting map for a sequence $\psi^{(N)} \in C({X^{\mathbb{N}}})$ if 
	\begin{equation}
	\label{eq31a}
	\psi^\infty({{a}}) = \lim_{N \to \infty} \psi^{(N)}({{a}}_N)
 	\end{equation}
for any converging sequence ${{a}}_N \to {{a}}$ in ${X^{\mathbb{N}}}$.
\end{definition}

The type of convergence in Definition~\ref{def_lim} is weaker than the strong convergence with respect to the supremum metric, but it is stronger than the pointwise convergence. One can verify that convergence (\ref{eq31a}) implies $\psi^\infty \in C({X^{\mathbb{N}}})$. The importance of condition (\ref{eq31a}) is that it is compatible with the composition in the RG operator (\ref{eq30R}): the limit $\psi^{(N)} \to \psi^\infty$ yields $\psi^{(N)} \circ \psi^{(N)} \to \psi^\infty  \circ \psi^\infty$. 

Let us assume that the limit $\psi^{(N)} \to \psi^\infty$ exists. Taking inviscid limit in both sides of (\ref{eq30}), one can see that the map $\psi^\infty$ is a fixed point of the RG operator:
	\begin{equation}
	\label{eq31}
	\psi^\infty = \mathcal{R}_g[\psi^\infty].
 	\end{equation}
This fixed point characterizes the inviscid limit of regularized solutions as follows.

\begin{theorem}[Inviscid limit]
\label{th3}
Let $\psi^\infty \in C({X^{\mathbb{N}}})$ be a fixed-point of the RG operator providing the limit $\psi^{(N)} \to \psi^\infty$. Then, for any initial conditions ${a} = (a_1,a_2,\ldots)$ and boundary conditions $(b_0,b_1,\ldots)$, the regularized solutions converge for any $(n,t) \in \mathcal{L}$ to
	\begin{equation}
	\label{eq32}
	{u}^{\infty}_n(t) = \lim_{N \to \infty} {u}^{(N)}_n(t),
 	\end{equation}
which is a solution of the ideal system (\ref{eq4})--(\ref{eq4b}). At integer times, this solution has the form
	\begin{equation}
	\label{eqKn4}
	{{u}}^{\infty}(t) = \phi^\infty_{b_{t-1}} \circ \cdots \circ \phi^\infty_{b_1} \circ \phi^\infty_{b_0}({a}),\quad
	\phi^\infty_b = \psi^{\infty} \circ \psi^{\infty}+\beta_b, \quad
	t \in \mathbb{N}.
 	\end{equation}
At non-integer times, the solution is given by Proposition~\ref{prop8} in the Appendix.
\end{theorem}

\begin{proof}
Expressions (\ref{eq19flowAt0}) and (\ref{eq19flowAt}) yield
	\begin{equation}
	\label{eqPr_s1}
	{{u}}^{(N)}(1) = \psi^{(N)} \circ \psi^{(N)}({a})+\beta_{b_0}({a}).
 	\end{equation}
Taking the limit $N \to \infty$ with $\psi^{(N)} \to \psi^\infty$, one derives 
	\begin{equation}
	\label{eqPr_s1lim}
	{{u}}^\infty(1) = 
	 \lim_{N \to \infty} {u}^{(N)}(1) 
	 = \psi^\infty \circ \psi^\infty({a})+\beta_{b_0}({a}),
 	\end{equation}
which coincides with (\ref{eq32}) and (\ref{eqKn4}) for $t = 1$. For any $t \in \mathbb{N}$, the analogous derivation follows from (\ref{eq19flowAtB}). Fractional times are analysed in Proposition~\ref{prop8} in the Appendix.

By Definition~\ref{def_det}, each particular equation of the ideal system in (\ref{eq4})--(\ref{eq4b}) is satisfied by the regularized solution ${u}_n^{(N)}(t)$ for sufficiently large $N$. Recall that the functions $f$ and $g$ in equation (\ref{eq4}) are continuous. It follows that all equations of the ideal system are satisfied in the inviscid limit (\ref{eq32}). 
\end{proof}
  
One of central questions for the inviscid limit is its dependence on the chosen regularization. Let us generalize the concept of regularization by replacing the simple cut-off rule (\ref{eq14}) in Definition~\ref{def_det} with an arbitrary uniquely defined dynamics at regularized scales $n > N$. We will see that Theorem~\ref{th3} can be used for proving the universality of inviscid limit, i.e., its independence of regularization. Considering the regularized variables at scales $n >  N$, we define their evolution over one turn-over time $t \mapsto t+\tau_N$ as
	\begin{equation}
	\label{eq14_gen0}
	\psi_{\mathrm{reg}}: \left({u}_{N+1}^{(N)}(t),{u}_{N+2}^{(N)}(t),\ldots\right) \mapsto \left({u}_{N+1}^{(N)}(t+\tau_N),{u}_{N+2}^{(N)}(t+\tau_N),\ldots\right),
 	\end{equation}
where $\psi_{\mathrm{reg}} \in C({X^{\mathbb{N}}})$ is a given regularization map. Relations (\ref{eq14_gen0}) corresponding to $n > N$ together with equations (\ref{eq4}) for $n \le N$ define the unique regularized solution $u_n^{(N)}(t)$ at times $t/\tau_N \in \mathbb{N}$. In particular, the map (\ref{eq19flowA}) for $N = 1$ from (\ref{eq19b}) is replaced by
	\begin{equation}
	\label{eq19b_gen}
	\psi^{(1)}: (a_1,a_2,\ldots) \mapsto \left(f(a_1,a_2),0,0,\ldots\right)+\sigma_-\circ \psi_{\mathrm{reg}} \circ \sigma_+(a).
 	\end{equation}
One can verify that the RG relations (\ref{eq30}) and (\ref{eq19flowAt}), as well as the RG operator (\ref{eq30R}) remain the same. Hence, the RG theory of Theorem~\ref{th3} extends to the new regularized solutions. 
We conclude that different regularizations affect only the initial map of the RG dynamics, while the RG operator itself is determined by couplings of the ideal system.

Theorem~\ref{th3} suggests that properties of the regularized solutions in the inviscid limit $N \to \infty$ are governed by the asymptotic dynamics (attractors) of the RG operator. 
In particular, let us assume that $\psi^{\infty}$ is a fixed-point attractor of the RG operator with the basin of attraction $B(\psi^\infty) \subset C({X^{\mathbb{N}}})$. Then we have the following universality property following from Theorem~\ref{th3}: the inviscid-limit solutions are independent of regularization for any $\psi^{(1)} \in B(\psi^\infty)$.

\section{Deterministic inviscid limit for a symbolic model}\label{sec_ex}

Let us show that symbolic models A and B introduced in Section~\ref{sec_symb_ex} exhibit different behaviors in the inviscid limit. As the next proposition shows, the model A with functions (\ref{eq5fg}) has a universal inviscid-limit solution given by a global attractor of the RG operator: the convergence $\psi^{(N)} \to \psi^\infty$ holds for any initial map $\psi^{(1)} \in C({X^{\mathbb{N}}})$.

\begin{proposition}\label{prog_fp}
The RG operator of the symbolic model A has the global fixed-point attractor 
	\begin{equation}
	\label{eq_exf_X1}
	\psi^{\infty} = \xi+\sum_{n = 3}^\infty{\zeta_n},\quad
	\zeta_n: {a} \mapsto \left\{\begin{array}{ll} 
	\mathbf{0}, & a_2 = \cdots = a_n = 0;\\[3pt]
	\mathbf{e}_n, & \textrm{otherwise};
	\end{array}\right.
 	\end{equation}
where $\mathbf{0} = (0,0,\ldots)$ and $\mathbf{e}_n \in X^\mathbb{N}$ is the  sequence with the unit $n$th component. 
\end{proposition}

We postpone the proof to the end of this section. Proposition~\ref{prog_fp} combined with Theorem~\ref{th3} assert that  the inviscid limit of regularization procedure defines a global-in-time solution of the ideal system for arbitrary initial and boundary conditions. This solution is universal: it does not depend on a choice of regularization. 
Figure~\ref{fig6} presents an example, which corresponds to the pre-blowup solution of Fig.~\ref{fig2}(a) extended to post-blowup times $t > T$. We remark that our symbolic system can be seen as a discrete analogue of the Burgers equation, for which a unique (shock wave) solution is obtained in the limit of vanishing viscosity~\cite{dafermos2005hyperbolic}.

\begin{figure}[tp]
\centering
\includegraphics[width=1\textwidth]{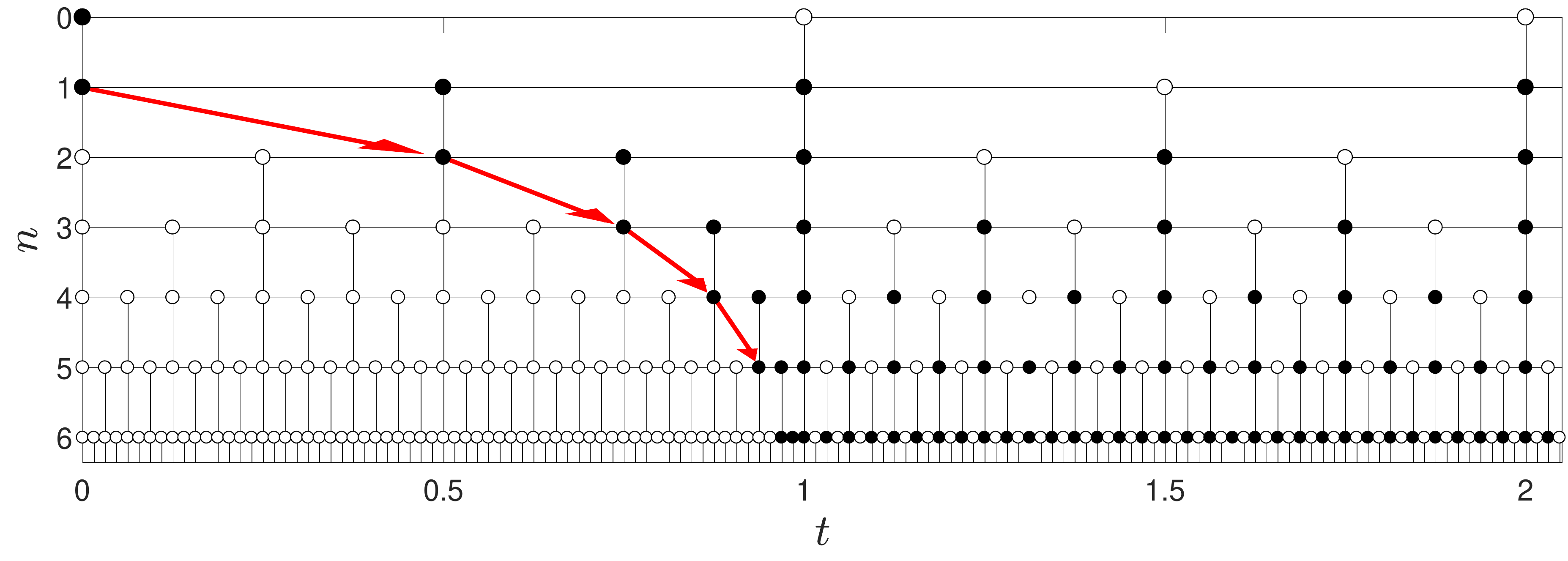}
\caption{Solution ${u}^\infty_n(t)$ of the symbolic model A, which has a finite time blowup at $T = 1$ shown by the  arrows. We consider the initial conditions ${a} = (1,0,0,0,\ldots)$ and boundary conditions $(b_0,b_1,b_2) = (1,0,0)$. Black circles correspond to ${u}^\infty_n(t) = 1$ and empty circles to ${u}^\infty_n(t) = 0$. This solution is obtained in the inviscid limit of regularization procedure.}
\label{fig6}
\end{figure}

A graphical picture of convergence is obtained by representing states ${{a}} = ({a}_1,{a}_2,\ldots) \in {X^{\mathbb{N}}}$ with elements $x \in \mathcal{C}$ from the middle-third Cantor set~\cite[\S1.9]{katok1997introduction}. The relation follows from the ternary (base 3) representation as 
	\begin{equation}
	\label{eq18b}
	x = 0.c_1c_2c_3\ldots,\quad 
	c_n = 	\left\{\begin{array}{cc}
	2,& {a}_n = 1;\\[2pt]
	0,& {a}_n = 0;
	\end{array}\right.
 	\end{equation}
where nonzero elements ${a}_n$ correspond to digits $2$ and no digits $1$ are present. Then, the discrete product topology in ${X^{\mathbb{N}}}$ corresponds to the topology induced by the usual distance $|x-x'|$ between points in $\mathcal{C} \subset \mathbb{R}$. In the Cantor set representation, the maps $\psi^{(N)}: {X^{\mathbb{N}}} \to {X^{\mathbb{N}}}$ are uniquely associated with the maps $\psi^{(N)}: \mathcal{C} \to \mathcal{C}$. Figure~\ref{fig5} shows the graphs of $\psi^{(N)}(x)$ for $N = 1,2,3,4$ and $10$, where the first map corresponds to (\ref{eq19b}). This figure conveniently visualizes the convergence of the RG dynamics to the fixed-point attractor as $N \to \infty$.
 
\begin{figure}[tp]
\centering
\includegraphics[width=0.45\textwidth]{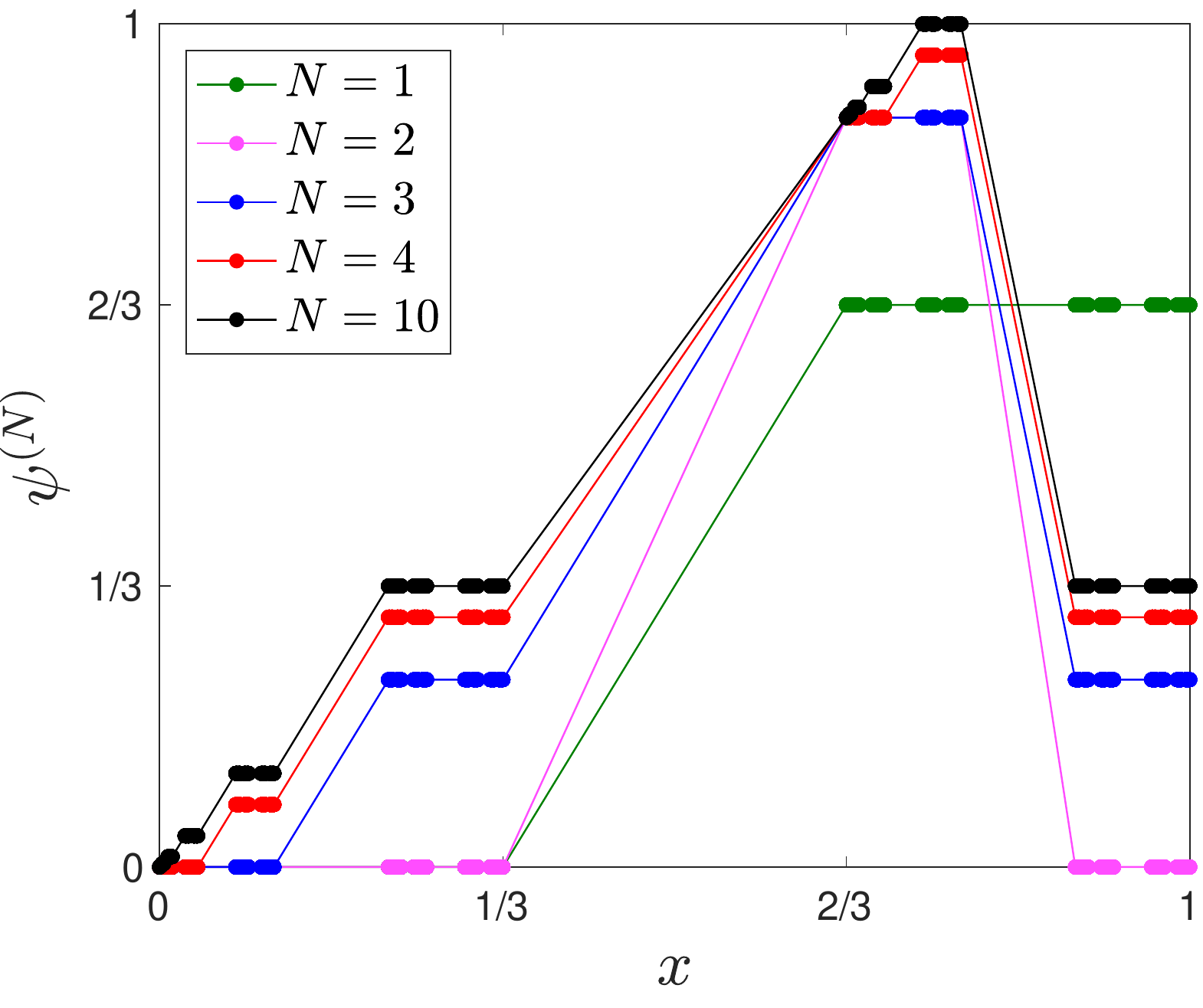}
\caption{The maps $\psi^{(N)}(x)$ in the Cantor set representation given by $N-1$ iterations of the RG operator for the model A. Thin lines connect Cantor set points for better visualization. With the increase of $N$, the functions converge to a fixed-point attractor.}
\label{fig5}
\end{figure}
\begin{figure}[tp]
\centering
\includegraphics[width=0.9\textwidth]{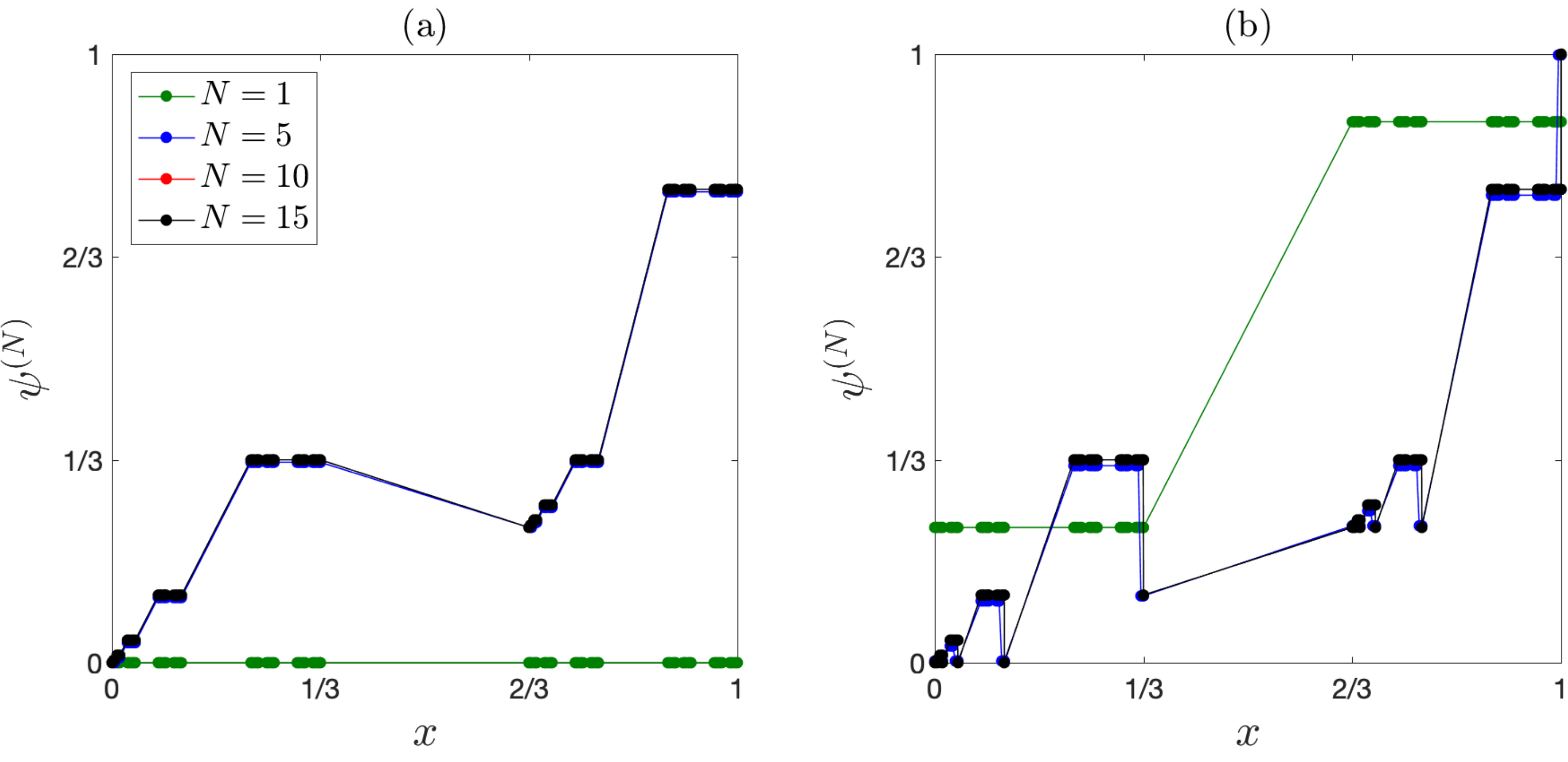}
\caption{The maps $\psi^{(N)}(x)$ in the Cantor set representation for the model B with two different regularizations: (a) $\psi^{(1)}(x) \equiv 0$ and (b) $\psi^{(1)}(x) \equiv 2/3$. The maps for $\psi^{(10)}$ and $\psi^{(15)}$  are very close (visually undistinguishable) pointing to the inviscid limit. However, the limits in the cases (a) and (b) are different. Also, one can see from the jump at $x = 1/3$ that the limiting function is discontinuous in the case (b).}
\label{fig5b}
\end{figure}

For the model B with functions (\ref{eq13}), the convergence is observed numerically in the case of regularization (\ref{eq14}); see Fig.~\ref{fig5b}(a). However, this is not true for other regularizations. For example, let us replace the regularization condition (\ref{eq14}) by
	\begin{equation}
	\label{eq14_exB}
	{u}_{N+1}^{(N)}(t) = 1, \quad {u}_n^{(N)}(t) = 0, \quad n > N+1, \quad t \ge 0,
 	\end{equation}
which assigns unit values at the scale $\ell_{N+1}$. In this case, 
	\begin{equation}
	\label{eq14_exB1}
	\psi^{(1)}: {a} \mapsto (f(a_1,1),1,0,0,\ldots).
 	\end{equation}
Numerical results shown in Fig.~\ref{fig5b}(b) suggest that the maps $\psi^{(N)}$ converge pointwize to a different limit. The limit $\psi^{(N)} \to \psi^\infty$ in the sense of Definition \ref{def_lim} does not exist, because the limiting map is discontinuous. We will show in Section~\ref{sec8_2} that the model A has the same inviscid limit solution in the presence of infinitesimal noise, while the model B becomes spontaneously stochastic.

\begin{proof}[Proof of Proposition~\ref{prog_fp}]
In this technical proof, we analyze sequentially RG iterations (\ref{eq30}) for an arbitrary initial map $\psi^{(1)} \in C({X^{\mathbb{N}}})$. 
For the first iteration, expression (\ref{eq30R}) yields 
	\begin{equation}
	\label{eq_exf_X3}
	\psi^{(2)} = \xi+\varepsilon^{(2)}, \quad \varepsilon^{(2)} = \sigma_- \circ \psi^{(1)} \circ \psi^{(1)} \circ \sigma_+.
 	\end{equation}
If follows from the second relation in (\ref{eq30Rs}) that $\varepsilon^{(2)}({a})$ has the vanishing first component. Similarly, we obtain 
	\begin{equation}
	\label{eq_exf_X4}
	\psi^{(3)} = \xi+\varepsilon^{(3)}, \quad
	\varepsilon^{(3)} = \sigma_- \circ \left(\xi+\varepsilon^{(2)}\right) \circ \left(\xi+\varepsilon^{(2)}\right) \circ \sigma_+.
 	\end{equation}
Let ${{u}} = \psi^{(3)}({a})$ and ${{u}}' = \varepsilon^{(3)}({a})$ for arbitrary ${a} \in {X^{\mathbb{N}}}$.  Using relations (\ref{eq30Rs}) and considering all possible values of $\xi$ from (\ref{eq30C}) and (\ref{eq5fg}), one can check that either $({u}'_1,{u}'_2) = (0,0)$, or $({u}'_1,{u}'_2) = (0,1)$ with $(a_2,a_3) = (1,0)$. In both cases, $({u}_1,{u}_2) \ne (1,0)$. Let us write the next iteration as
	\begin{equation}
	\label{eq_exf_X5}
	\psi^{(4)} = \xi+\varepsilon^{(4)}, \quad
	\varepsilon^{(4)} = \sigma_- \circ \left(\xi+\varepsilon^{(3)}\right) \circ \psi^{(3)} \circ \sigma_+.
 	\end{equation}
Using relations (\ref{eq30C}), (\ref{eq5fg}), (\ref{eq30Rs}) and the properties of $\psi^{(3)}$ and $\varepsilon^{(3)}$, one can verify that the first two components of $\varepsilon^{(4)}({a})$ vanish identically.

It remains to show that all further iterations take the form
	\begin{equation}
	\label{eq_exf_X6}
	\psi^{(N)} = \xi+\sum_{n = 3}^{N-2}{\zeta_n}+\varepsilon^{(N)}, \quad N \ge 4,
 	\end{equation}
where the first $N-2$ components of the map $\varepsilon^{(N)}({a})$ vanish identically. One can see that this expression implies the convergence (\ref{eq31a}) to the fixed-point attractor (\ref{eq_exf_X1}).

We already proved relation (\ref{eq_exf_X6}) for $N = 4$. Assuming that it holds for some $N \ge 4$, we express
	\begin{equation}
	\label{eq_exf_X7}
	\begin{array}{rl}
		\psi^{(N+1)} 
		&
		= \mathcal{R}_g[\psi^{(N)}] 
		= \xi+\sigma_- \circ \psi^{(N)} \circ \psi^{(N)} \circ \sigma_+ 
		\\[7pt] 
		& \displaystyle
		= \xi+\sigma_- \circ \xi \circ \psi^{(N)} \circ \sigma_+ 
		+\left(\sum_{n = 3}^{N-2} 
		\sigma_- \circ \zeta_n \circ \psi^{(N)} \circ \sigma_+\right)
		+\sigma_- \circ \varepsilon^{(N)} \circ \psi^{(N)} \circ \sigma_+,
	\end{array}
 	\end{equation}
where we substituted (\ref{eq_exf_X6}) for the first $\psi^{(N)}$.
Using explicit forms of the maps $\xi$ and $\zeta_n$ and expression (\ref{eq_exf_X6}), it is straightforward to verify that
	\begin{equation}
	\label{eq_exf_X8}
	\sigma_- \circ \xi \circ \psi^{(N)} \circ \sigma_+ 
	= \sigma_- \circ \xi \circ \left(\xi+\sum_{n = 3}^{N-2}{\zeta_n}+\varepsilon^{(N)}\right) \circ \sigma_+
	= \zeta_3, 
 	\end{equation}
	\begin{equation}
	\label{eq_exf_X8b}
	\sigma_- \circ \zeta_n \circ \psi^{(N)} \circ \sigma_+ = 
	\sigma_- \circ \zeta_n \circ \left(\xi+\sum_{n = 3}^{N-2}{\zeta_n}+\varepsilon^{(N)}\right) \circ \sigma_+ = \zeta_{n+1}
 	\end{equation}
for any $n = 3,\ldots,N-2$. Expressions (\ref{eq_exf_X7})--(\ref{eq_exf_X8b}) yield the representation (\ref{eq_exf_X6}) for $N+1$ with $\varepsilon^{(N+1)} = \sigma_- \circ \varepsilon^{(N)} \circ \psi^{(N)} \circ \sigma_+$. Due to the map $\sigma_-$ from (\ref{eq30Rs}), the first $N-1$ components of the map $\varepsilon^{(N+1)}$ vanish. This concludes the proof by induction. We remark that, according to (\ref{eq_exf_X6}), the first $N-2$ components of the maps $\psi^{\infty}$ and $\psi^{(N)}$ coincide.
\end{proof}

\section{Stochastically regularized solutions}\label{sec8}

Now let us consider the case when the sequence of regularized flow maps $\psi^{(N)}$ does not converge to a fixed point $\psi^\infty$ of the RG operator. 
For example, the RG dynamics can be chaotic and, therefore, require a probabilistic formulation for the inviscid limit $N \to \infty$. Such formulation follows naturally in fluid dynamics, where viscous forces at small scales coexist with microscopic (e.g., thermal) fluctuations~\cite{ruelle1979microscopic,landau1987fluid,eyink2021dissipation}. This motivates a definition of the stochastically regularized system below: we modify the previous Definition~\ref{def_det} by adding random perturbations (noise) at the regularized scale $\ell_{N+1}$. This perturbation uses a sequence ${x}_0,{x}_1,\ldots$ of independent and identically distributed (i.i.d.) random variables with values in $X$  and probability measure $\mu_x$. 

\begin{definition} \label{def_st}
Given a regularization scale $\ell_N$ with $N \in \mathbb{N}$, we introduce the \textit{stochastically regularized solution} as a collection of random variables $\mathfrak{u}_n^{(N)}(t)$ on the lattice $(n,t) \in \mathcal{L}$. For $n > N$, we define
	\begin{equation}
	\label{eqSK1}
	\mathfrak{u}_{n}^{(N)}(t) = \left\{\begin{array}{ll}
	{x}_{t/\tau_n}, & n = N+1,\ t \ge 0; \\[3pt]
	0, & n > N+1,\ t \ge 0.
	\end{array}\right.
 	\end{equation}
The remaining variables $\mathfrak{u}_{n}^{(N)}(t)$, $n \le N$, are given by equations (\ref{eq4}) with initial conditions (\ref{eq4c}) for $n = 1,\ldots,N$ and boundary conditions (\ref{eq4b}). 
\end{definition}

Random variables $\mathfrak{u}_{n}^{(N)}(t)$ in Definition~\ref{def_st} are expressed iteratively as functions of initial and boundary conditions and random variables $x_m$. These functions determine probability distributions of variables $\mathfrak{u}_{n}^{(N)}(t)$ in terms of the distributions of $x_m$.  In this section, we focus on solutions $\mathfrak{u}^{(N)}(t) = (\mathfrak{u}_{1}^{(N)}(t),\mathfrak{u}_{2}^{(N)}(t),\ldots)$ at half-integer and integer times $t$. Random variables $\mathfrak{u}^{(N)}(t)$ are valued in the measurable space $(X^\mathbb{N},\mathcal{A})$, where $\mathcal{A}$ is the Borel $\sigma$-algebra of measurable sets in the space $X^{\mathbb{N}}$ with the product topology.
Kolmogorov's extension theorem defined the unique probability measure for each random variable $\mathfrak{u}^{(N)}(t)$, which we denote as $P^{(N)}_t(A|a)$ indicating the dependence on initial condition $a$. It defines the probability of $\mathfrak{u}^{(N)}(t) \in A$ for any $A \in \mathcal{A}$.

One can see that $P^{(N)}_t(A|a)$ is a probability (Markov) kernel in $(X^\mathbb{N},\mathcal{A})$. We now introduce operations with probability kernels (see, e.g.,~\cite{taylor2001user}), which replace respective operation with maps from the deterministic description.
We recall that a probability kernel $\Psi(A|a)$ is a function from $\mathcal{A} \times X^\mathbb{N}$ to the interval $[0,1]$, such that $\Psi(\cdot|a)$ is a probability measure for any $a \in X^\mathbb{N}$ and $\Psi(A|a)$ is a measurable function of $a$ for any $A \in \mathcal{A}$. 
For an arbitrary measure $\mu$ on ${X^{\mathbb{N}}}$ and probability kernel $\Psi$, one defines the measure $\Psi\mu$ as
	\begin{equation}
	\label{eqSS1}
	\quad \Psi\mu(A) = \int_{{X^{\mathbb{N}}}} \Psi(A|a)\mu(d{a}),\quad A \in \mathcal{A}.
 	\end{equation}
For two kernels $\Psi$ and $\Psi'$, their composition $\Psi' \circ \Psi$ is a kernel with 
	\begin{equation}
	\label{eqSK3}
	(\Psi' \circ \Psi)(A|a) = (\Psi' \mu)(A), \quad \mu = \Psi(\cdot|a),
 	\end{equation}
where we use the previous definition (\ref{eqSS1}). Measures (\ref{eqSK3}) describe a probability distribution for a composition $\psi' \circ \psi(a)$, where $\psi(a)$ and $\psi'(a)$ are statistically independent random variables with the distributions described by the kernels $\Psi$ and $\Psi'$.
Let us introduce the deterministic probability kernel $\Xi$ representing the coupling map $\xi(a)$ as the Dirac meaure $\Xi(\cdot|a) = \delta_{\xi(a)}$. 
We define the convolution kernel $\Psi \ast  \Xi$ as 
	\begin{equation}
	\label{eqSK5}
	(\Psi \ast  \Xi)(A|a) = \Psi(A_{\xi({a})}|a), \quad
	A_{\xi({a})} = \{{{u}} \in {X^{\mathbb{N}}}: {{u}}+\xi({a}) \in A\},
	\quad A \in \mathcal{A}.
 	\end{equation}
This expression follows from the conventional definition for the convolution of measures $\Psi_{{a}}$ and $\Xi$, which describes the probability distribution of the sum $\psi(a)+\xi({a})$ with a random variable $\psi(a)$. Similarly, we introduce the deterministic probability kernels with the Dirac measures $\Sigma_+(\cdot|a) = \delta_{\sigma_+({a})}$, $\Sigma_-(\cdot|a) = \delta_{\sigma_-({a})}$ and $B_b(\cdot|a) = \delta_{\beta_b({a})}$, which represent the shifts from (\ref{eq30Rs}) and the map from (\ref{eq19flowX}). 

We now describe the kernels $P_t^{(N)}$ following the same line of derivations as for the deterministic solutions $u^{(N)}(t)$ in Section~\ref{sec_DR}.
At time $t = 1$, the kernel ${P}^{(N)}_1$ depends on the boundary condition $b_0$;  see Fig.~\ref{fig1}. We introduce a specific notation for this kernel as $\Phi^{(N)}_b$ defined for any $b \in X$ such that
	\begin{equation}
	\label{eqKn1}
	P^{(N)}_1 = \Phi^{(N)}_{b_0}.
 	\end{equation}
Since the system is translation invariant with respect to unit-time steps, and recalling that random variables $x_m$ are statistically independent, transition from $t = 1$ to $t = 2$ is given by the kernel $\Phi^{(N)}_{b_1}$, etc. Hence, at any integer time we find
	\begin{equation}
	\label{eqKn2}
	{P}^{(N)}_t = \Phi^{(N)}_{b_{t-1}} \circ \cdots \circ \Phi^{(N)}_{b_0}, \quad
	t \in \mathbb{N}.
 	\end{equation}

Consider now the half-time kernel and denote it by
	\begin{equation}
	\label{eqKn3}
	\Psi^{(N)} = {P}^{(N)}_{\tau_1}.
 	\end{equation}
For $N = 1$, using relations (\ref{eqSK1}) in (\ref{eq4}) and (\ref{eq4c}), we find  ${\mathfrak{u}}^{(1)}(\tau_1) = \left(f(a_1,{x}_0),{x}_2,0,0,\ldots\right)$, where $x_0$ and ${x}_2$ are independent random variables with the probability measure $\mu_x$. The probability distribution of $f(a_1,{x}_0)$ is given by the push-forward measure $f(a_1,\cdot)_\sharp\mu_x$. This yields $\Psi^{(1)}$ as a family of measures
	\begin{equation}
	\label{eqSK2m}
	\Psi^{(1)}(\cdot|a) = f(a_1,\cdot)_\sharp\mu_x \times \mu_x \times \delta_0 \times \delta_0 \times \cdots.
 	\end{equation}

In the deterministic case, we derived relation (\ref{eq19flowAt}) for the maps $\phi^{(N)}_{b_0}$ and $\psi^{(N)}$. This derivation has a straightforward extension to the stochastically regularized solutions, where the maps $\phi^{(N)}_{b_0}$, $\psi^{(N)}$ and $\beta_{b_0}$ are substituted by the respective kernels $\Phi^{(N)}_{b_0}$, $\Psi^{(N)}$ and ${{B}_{b_0}}$ and the addition by the convolution. This yields
	\begin{equation}
	\label{eqK_e1M}
	\Phi^{(N)}_{b_0} = \left(\Psi^{(N)} \circ \Psi^{(N)}\right)*{{B}_{b_0}}.
 	\end{equation}
We now express $\Psi^{(N+1)}$ in terms of $\Psi^{(N)}$ by introducing a \textit{stochastic RG operator} acting on probability kernels. 

\begin{definition}
The stochastic RG operator maps probability kernels to probability kernels as
	\begin{equation}
	\label{eqS7}
	\mathfrak{R}_g: \Psi \mapsto \left(\Sigma_- \circ \Psi \circ \Psi \circ \Sigma_+\right) \ast \Xi.
 	\end{equation}
\end{definition}

\begin{proposition}
\label{prop6} 
For any $N \in \mathbb{N}$, the following iterative relation holds:
	\begin{equation}
	\label{eqS6}
	\Psi^{(N+1)} = \mathfrak{R}_g[\Psi^{(N)}].
 	\end{equation}
\end{proposition}

\begin{proof}
The proof follows exactly the same line of derivation as the proof of Proposition~\ref{prop3}, observing that the same relations are satisfied by random variables $\mathfrak{u}_{n}^{(N)}(t)$ of the stochastically regularized solution. For example, the scaling relation (\ref{eq5N}) is satisfied for the random variables $\tilde{\mathfrak{u}}_{n}^{(N)}(t) = \mathfrak{u}_{n+1}^{(N+1)}(t/2)$ with the initial and boundary conditions (\ref{eq5Nic_b}). The functions $\psi^{(N)}$ and $\phi_{b_0}^{(N)}$ mapping initial conditions to solutions $u^{(N)}(\tau_1)$ and $u^{(N)}(1)$ become the probability kernels $\Psi^{(N)}$ and $\Phi_{b_0}^{(N)}$ providing probability measures for random variables $\mathfrak{u}^{(N)}(\tau_1)$ and $\mathfrak{u}^{(N)}(1)$. Recall that compositions and sums of maps correspond to compositions and convolutions for kernels. 
For example, the next step in (\ref{eq19flowPrA0}) is written using (\ref{eqK_e1M}) as $\Phi^{({N})}_{\tilde{b}_0} = \big(\Psi^{({N})} \circ \Psi^{({N})}\big) * B_{\tilde{b}_0}$. Rewriting the other steps in a similar way, one arrives at the relations $\Psi^{(N+1)} = \left(\Sigma_- \circ \Psi \circ \Psi \circ \Sigma_+\right) \ast \Xi$, which is the probabilistic form of map equation (\ref{eq30P}).
\end{proof}

\section{Spontaneously stochastic solutions}\label{sec9}

Let us study the inviscid limit $N \to \infty$. Following~\cite{karr1975weak}, we introduce the following definition for the convergence of kernels $\Psi^{(N)}$, which is a stochastic version of Definition~\ref{def_lim}.

\begin{definition}\label{st_lim}
We say that $\Psi^{(N)} \to \Psi^\infty$ is a limiting kernel for a sequence $\Psi^{(N)}$ if 
	\begin{equation}
	\label{eqSlim}
	\int \varphi({{u}})\Psi^\infty(d{{u}}|a) = 
	\lim_{N \to \infty} \int \varphi({{u}})\Psi^{(N)}(d{{u}}|a_N)
 	\end{equation}
holds for any converging sequence ${a}_N \to {a}$ in ${X^{\mathbb{N}}}$ and any bounded and uniformly continuous function $\varphi \in C({X^{\mathbb{N}}})$.
\end{definition}

Considering the limit $N \to \infty$ in both sides of relation (\ref{eqS6}), the convergence $\Psi^{(N)} \to \Psi^\infty$ implies that the limiting kernel
	\begin{equation}
	\label{eqSS2b}
	\Psi^\infty = \mathfrak{R}_g[\Psi^\infty]
 	\end{equation}
is a fixed point of the stochastic RG operator. This fixed-point probability kernel yields a complete description of the inviscid limit for stochastically regularized solutions as follows.

Consider the full stochastically regularized solution $\{\mathfrak{u}_n^{(N)}(t)\}_{(n,t) \in \mathcal{L}}$ as a stochastic process. It represents a collection of random variables taking values in the product space $X^\mathcal{L}$ with the product topology and the corresponding Borel $\sigma$-algebra of measurable sets. Given the boundary and initial conditions, Kolmogorov's extension theorem defines the unique probability measure $\mathcal{P}^{(N)}$ on $X^\mathcal{L}$ for the stochastic process $\{\mathfrak{u}_n^{(N)}(t)\}_{(n,t) \in \mathcal{L}}$. Notice that $\mathcal{P}^{(N)}$ depends on initial and boundary conditions, but we do not specify this explicitly in the notation.

\begin{theorem}[Spontaneous stochasticity]
\label{th3ss}
Let $\Psi^\infty$ be a fixed-point kernel of the stochastic RG operator providing the limit $\Psi^{(N)} \to \Psi^\infty$. Then, for any initial conditions ${a}=(a_1,a_2,\ldots)$ and boundary conditions $(b_0,b_1,\ldots)$,  probability measures $\mathcal{P}^{(N)}$ of stochastically regularized solutions converge weakly to a measure
	\begin{equation}
	\label{eq32stoch}
	\mathcal{P}^{\infty} = \lim_{N \to \infty} \mathcal{P}^{(N)}.
 	\end{equation}
The measure $\mathcal{P}^{\infty}$ is supported on a set of (weak) solutions of the ideal system (\ref{eq4})--(\ref{eq4b}). At integer times, statistics (\ref{eq32stoch}) define a Markov chain with the kernels
	\begin{equation}
	\label{eq41}
	{P}^{\infty}_t = \Phi^{\infty}_{b_{t-1}} \circ \cdots \circ \Phi^{\infty}_{b_0},\quad
	\Phi^\infty_{b} = \left(\Psi^{\infty} \circ \Psi^\infty\right)*{{B}_{b}}, \quad
	t \in \mathbb{N}.
 	\end{equation}
Probability distributions at non-integer times are described by Proposition~\ref{prop8stat} in the Appendix. 
\end{theorem}
\begin{proof}
Theorem~4 from \cite{karr1975weak} states that the convergence $\Psi^{(N)} \to \Psi^\infty$ in Definition~\ref{st_lim} implies the weak convergence of probability measures for respective Markov chains. We showed that the stochastically regularized solution $\mathfrak{u}^{(N)}(t)$ at integer times is a Markov chain with probability kernels given by (\ref{eqKn2}) and (\ref{eqK_e1M}). Hence, the convergence (\ref{eq32stoch}) with relations (\ref{eq41}) hold for probability measures $\mathcal{P}^{(N)}$ considered (projected) at integer times. The full stochastically regularized solution includes non-integer times, which is the intrinsic property of our lattice $\mathcal{L}$. As we show in Proposition~\ref{prop8stat} in the Appendix, solutions at non-integer times $t = m\tau_n$ can be written in terms of compositions of kernels $\Psi^{(N-n+1)},\ldots,\Psi^{(N)}$. As a consequence, one can verify that the proof in \cite{karr1975weak} has a straightforward extension to the full probability measures $\mathcal{P}^{(N)}$. 

It remains to prove that the measure $\mathcal{P}^{(N)}$ is supported on a set of weak solutions of the ideal system for the respective initial and boundary conditions. Let $E_N \subset X^\mathcal{L}$ be a subset containing the fields $\{{u}_n(t)\}_{(n,t) \in \mathcal{L}} \in X^{\mathcal{L}}$ satisfying equations (\ref{eq4})--(\ref{eq4b}) of the ideal system at all scales $n \le N$. This yields a nested sequence of subsets with $E_1 \supset E_2 \supset E_3 \cdots$.
By Definition~\ref{def_st}, the measure $\mathcal{P}^{(N)}$ is supported on $E_N$. Hence, the weak limit $\mathcal{P}^{\infty}$ of these measures is supported on the intersection $\bigcap_{n \in \mathbb{N}} E_N$, which is a set of all weak solutions of the ideal system.
\end{proof}

Theorem~\ref{th3ss} states that the uncertainty can persist despite the stochastic regularization is completely removed in the inviscid limit. It is manifested in a stochastic process determined by the limiting measures $\mathcal{P}^{\infty}$. Each realization of this process (with probability 1) is a solution of the deterministic ideal system with deterministic initial and boundary conditions. We now give the formal definition and criterion for the spontaneous stochasticity. 

\begin{definition}
\label{def_sps}
Ideal system (\ref{eq4})--(\ref{eq4b}) is called \textit{spontaneously stochastic} if for some initial and boundary conditions the weak inviscid limit (\ref{eq32stoch}) exists with a nontrivial (not deterministic) distribution $\mathcal{P}^{\infty}$.
\end{definition}

\begin{corollary}
\label{cor_sps}
Ideal system (\ref{eq4})--(\ref{eq4b}) is spontaneously stochastic if the limit $\Psi^{(N)} \to \Psi^\infty$ exists and, for some $a \in X^\mathbb{N}$, the probability measure $\Psi^\infty(\cdot|a)$ is nontrivial.
\end{corollary}


Let $B(\Psi^\infty)$ be a set of probability kernels $\Psi$ such that $\mathfrak{R}_g^N[\Psi] \to \Psi^\infty$, which can be regarded as a basin of attraction.
The immediate consequence is the universality of spontaneously stochastic solutions with respect to regularization: Theorem~\ref{th3ss} is valid and yields the same limiting solution for any stochastic regularization given by $\Psi^{(1)} \in B(\Psi^{\infty})$.  

We see that the phenomenon of spontaneous stochasticity is related to the theory of dynamical systems, which studies convergence to invariant measures at large times \cite{katok1997introduction}. One may expect the convergence $\Psi^{(N)} \to \Psi^{\infty}$ when the deterministic RG dynamics (\ref{eq30}) is chaotic. Another scenario corresponds to multi-stability, when the deterministic RG operator has several attractors. In this case, the limit $\Psi^{\infty}$ may depend on how the initial kernel $\Psi^{(1)}$ is distributed among the basins of those attractors.

\section{Examples of spontaneously stochastic models}\label{sec_sp_ex}

\subsection{Multi-scale interacting phases with expanding couplings}\label{sec8_1}

Let us consider a system from to a class of models introduced in our previous work~\cite{mailybaev2021spontaneously}, for which the inviscid limit can be studied analytically. We take the space $X = \mathbb{S}^1 = \mathbb{R}/\mathbb{Z}$ as the circle group  with operations modulo $1$ and the functions
	\begin{equation}
	\label{eqAC1}
	f({u},{u}') = 2{u}+2{u}' \ \ (\mathrm{mod}\ 1),\quad
	g({u},{u}') \equiv 0,
 	\end{equation}
reducing equations (\ref{eq4}) to the form
	\begin{equation}
	\label{eqAC2}
	{u}_n(t) = 2{u}_n(t-\tau_n)+2{u}_{n+1}(t-\tau_n)\ \ (\mathrm{mod}\ 1).
 	\end{equation}
Thus, the ideal dynamics represents the doubling of phases at each scale and adding the same contribution from next smaller scales; see Fig.~\ref{fig1}.
We note that relation (\ref{eqAC2}) does not involve the scale number $n-1$ and, therefore, the boundary conditions (\ref{eq4b}) can be neglected. We study the inviscid limit $N \to \infty$ for given initial conditions (\ref{eq4c}) and the stochastic regularization (\ref{eqSK1}) with random phases ${x}_m$ taking values in $\mathbb{S}^1$. 

We will show that, because of the expanding coupling, the probability distributions of random phases $x_m$ are pushed-forward by the dynamics to uniform measures on $\mathbb{S}^1$ for most variables $\mathfrak{u}_n^{(N)}(t)$ in the inviscid limit. The exceptions converge to Dirac measures and correspond to the variables $\mathfrak{u}_n^{(N)}(\tau_n)$ at one turn-over time, because these variables do not depend on $x_m$. Thus, the inviscid limit yields the universal (independent of regularization) spontaneously stochastic solution.

\begin{proposition}\label{prop_exS1}
If i.i.d. random variables ${x}_m$ are absolutely continuous, then the stochastically regularized solutions have a weak inviscid limit at any time on the lattice as
	\begin{equation}
	\label{eqAC6}
	\lim_{N \to \infty}{P}^{(N)}_{{t}}(\cdot|a) 
	= {P}^\infty_{{t}}(\cdot|a)
	= \left\{\begin{array}{ll}
	\delta_{f(a_n,a_{n+1})} \times \mu \times \mu \times \cdots,& t  = \tau_n; \\[3pt] 
	\mu \times \mu \times \mu \times \cdots,& \textrm{otherwise};
	\end{array}\right. 
 	\end{equation}
where $\mu$ is a uniform probability measure on $\mathbb{S}^1$. 
\end{proposition}

\begin{proof}
This proof largely relies to the results of our previous work~\cite{mailybaev2021spontaneously}, and we present it here in a short form. Since relations (\ref{eqAC2}) are linear, each variable $\mathfrak{u}_n^{(N)}(t)$ of the stochastically regularized solution is a linear combination of a finite number of initial phases $a_n \in \mathbb{S}^1$ and random variables ${x}_m$. For $t = \tau_1$, we have
	\begin{equation}
	\label{eqAC3}
	\mathfrak{u}_1^{(N)}(\tau_1) = f(a_1,a_2),\quad
	\mathfrak{u}_n^{(N)}(\tau_1) = 
	p_n^{(N)}{x}_0+f_n^{(N)}({a},x) \ \ (\mathrm{mod}\ 1), \quad
	n = 2,\ldots,N,
 	\end{equation}
where we separated the contribution of the first random variable ${x}_0$ from contribution of initial conditions ${a}$ and other random variables $x = ({x}_1,{x}_2,\ldots)$. One can show the following

\begin{lemma} 
Coefficients $p_n^{(N)}$ are positive integers with the properties
	\begin{equation}
	\label{eqAC4}
	\lim_{N \to \infty} p_2^{(N)} = \infty,\quad 
	\lim_{N \to \infty} \frac{p_{n+1}^{(N)}}{p_n^{(N)}} = \infty,\quad n = 2,\ldots,N-1.
 	\end{equation}
\end{lemma}
The proof of Lemma is given in \cite[Lemma~2]{mailybaev2021spontaneously}, where we considered a similar problem with $X$ being a torus (two phases).

The measure ${P}^{(N)}_{{\tau_1}}(\cdot|a)$ defines a probability distribution of $\mathfrak{u}^{(N)}(\tau_1)$ for a given initial condition ${a}$. 
In our previous work~\cite[Lemma~1]{mailybaev2021spontaneously}, we used the Fourier expansion technique under a more general setting. This result can be directly applied for  proving the limit (\ref{eqAC6}). 
\end{proof}

By Proposition~\ref{prop_exS1}, the model (\ref{eqAC2}) is spontaneously stochastic and universal: the solution of deterministic ideal system obtained in the inviscid limit is stochastic and independent of distributions of random variables $x_m$. 
We refer an interested reader to our previous work \cite{mailybaev2021spontaneously} for numerical tests of convergence to the spontaneously stochastic limit. In particular, we argued that this convergence is double exponential in $N$ for most of the variables, which means that the spontaneous stochasticity can be observed at moderately large $N$ even when random phases $x_m$ are extremely small. 

One may recognize a similarity between this example and the spontaneous stochasticity in the context of Kelvin--Helmholtz instability studied numerically in \cite{thalabard2020butterfly}. The latter system was considered with infinitesimal small-scale perturbations of initial vorticity. In this analogy, one can interpret our variable ${u}_n(t)$ as representing a deviation of the flow from the stationary solution.

We would also like to mention that dynamical systems on lattices ($\mathbb{Z}$ or $\mathbb{Z}^2$) with expanding couplings, the so-called coupled-map lattices (CML), have been studied previously; see e.g. \cite{bricmont1996high,chazottes2005dynamics}. But our lattice has a different space-time structure: it cannot be seen as a classical dynamical system and, thus, requires new techniques for the analysis. 

\subsection{Spontaneous stochasticity and digital turbulence}\label{sec8_2}

First, let us consider the symbolic model A with functions (\ref{eq5fg}) studied in Sections \ref{sec_symb_ex} and \ref{sec_ex}, but now imposing the stochastic regularization (\ref{eqSK1}). We consider  i.i.d. random variables ${x}_m$ given by the Bernoulli distribution with $p = 1/2$.  Considering two different samples of random variables, we computed the corresponding maps $\psi^{(N)}({a})$ and their Cantor set representations $\psi^{(N)}(x)$, which are presented in Fig.~\ref{fig7} for $N = 2$, $4$ and $6$.  Comparing with Fig.~\ref{fig5} verifies numerically that the stochastic regularization yields the same deterministic solution in the inviscid limit: the limiting probability kernel $\Psi^\infty$ consists of Dirac measures $\Psi^\infty(\cdot|a) = \delta_{\psi^\infty({{a}})}$ with $\psi^\infty$ from (\ref{eq_exf_X1}).

\begin{figure}[tp]
\centering
\includegraphics[width=1\textwidth]{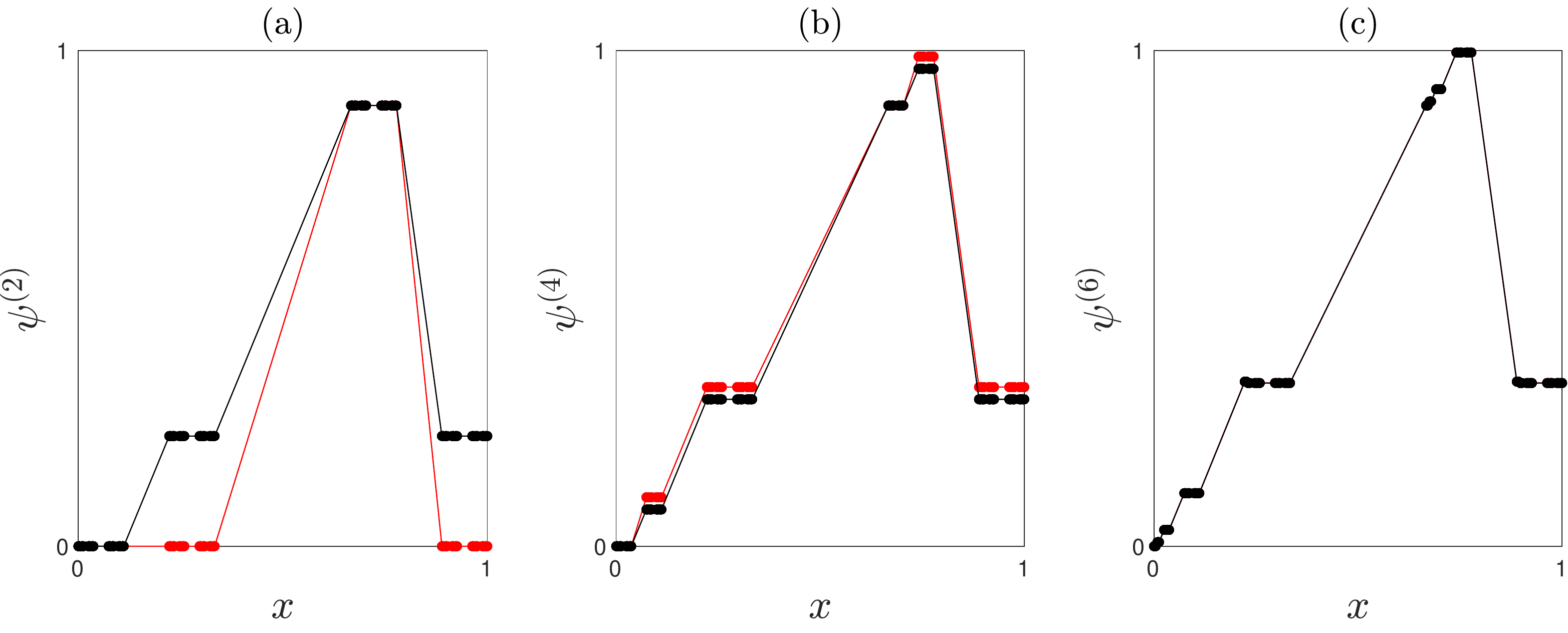}
\caption{Each panel shows two random samples (red and black) of the stochastically regularized map $\psi^{(N)}(x)$ in the Cantor set representation for the symbolic model A: (a) $N =2$, (b) $N = 4$ and (c) $N = 6$. Comparing with Fig.~\ref{fig5}, one observes the convergence to a deterministic limit.}
\label{fig7}
\end{figure}

Let us now study the same stochastic regularization for the symbolic model B with functions (\ref{eq13}). We will see that this model is spontaneous stochastic and demonstrates the ``digital turbulence'': irregular and unpredictable multi-scale dynamics of alternating laminar and turbulent ($0$ and $1$) states.  
Two random samples of maps $\psi^{(N)}(x)$ are shown in Figs.~\ref{fig8}(a,b) for $N = 7$ and $10$. Unlike Fig.~\ref{fig7}, these samples do not converge for large $N$. Therefore, we do not expect a deterministic inviscid limit. The same functions but now for $100$ random samples are shown in Figs.~\ref{fig8}(c,d). These plots characterize the support of the probability measures $\Psi^{(N)}(\cdot|{a})$, and their close similarity for $N = 7$ and $N = 10$ is our first numerical evidence of statistical convergence.

\begin{figure}[tp]
\centering
\includegraphics[width=0.8\textwidth]{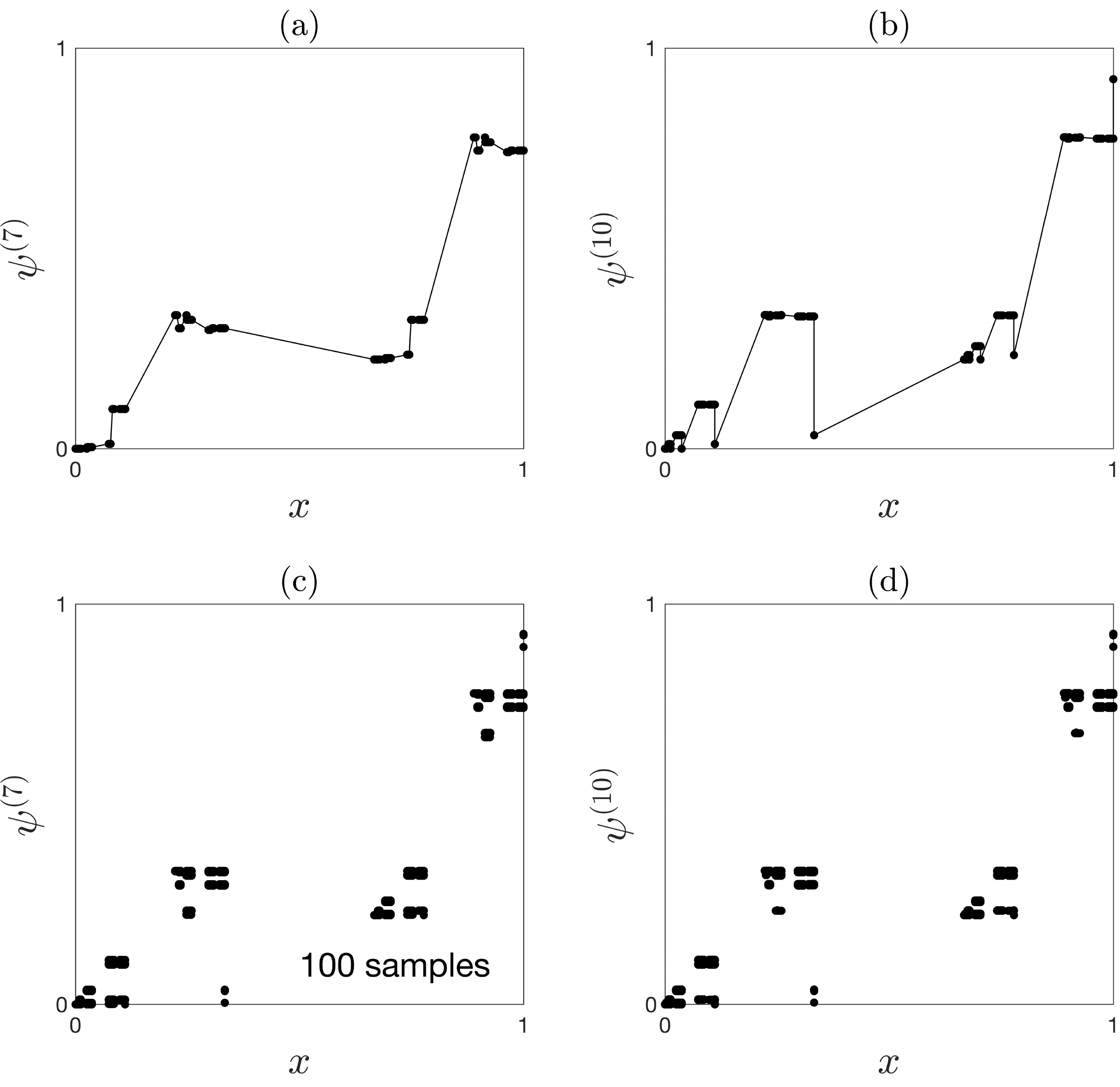}
\caption{Functions $\psi^{(N)}(x)$ in the Cantor set representation for the stochastic regularization of the symbolic model B. Observe that random samples for (a) $N = 7$ and (b) $N = 10$ are different, while one hundred random samples for (c) $N = 7$ and (d) $N = 10$ demonstrate the statistical convergence. }
\label{fig8}
\end{figure}
\begin{figure}[htp]
\centering
\includegraphics[width=0.95\textwidth]{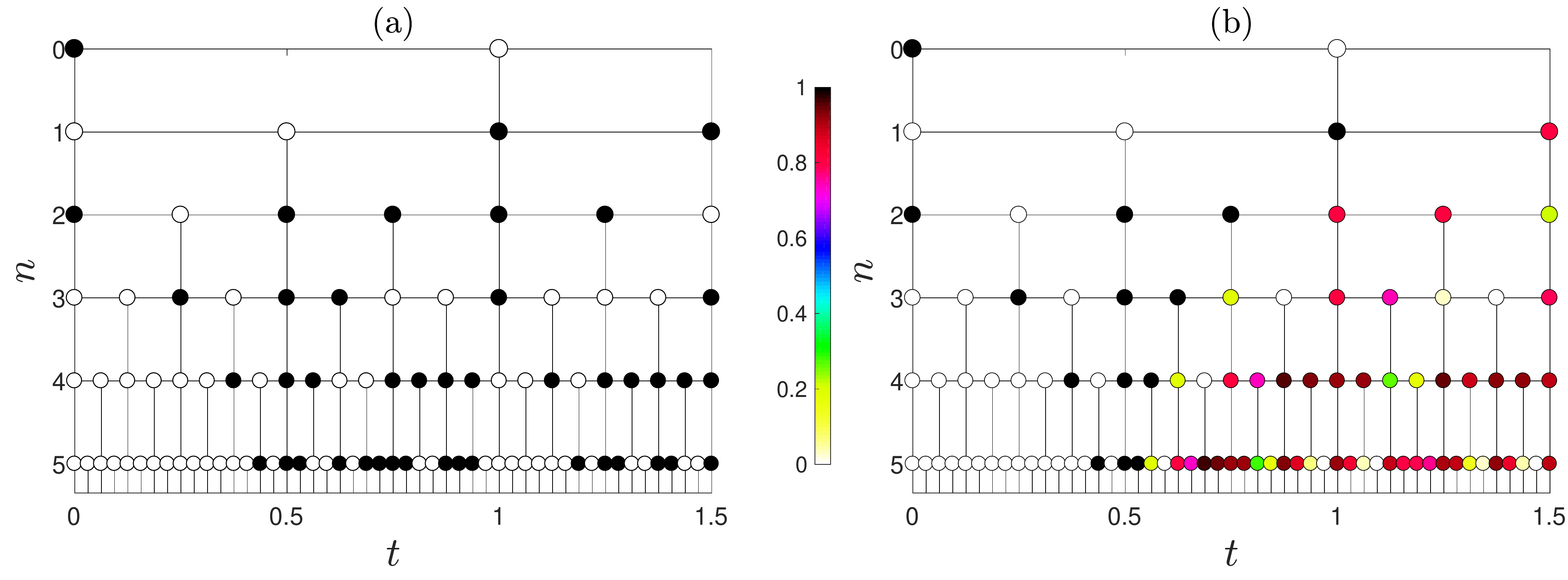}
\caption{Stochastically regularized solutions of the model B with $N = 20$, initial condition ${a} = (0,1,0,0,0,\ldots)$ and boundary conditions $(b_0,b_1) = (1,0)$. (a) A random sample, where black circles correspond to ${u}_n^{(N)}(t) = 1$ and white circles to $0$. (b) Mean values of random variables ${u}_n^{(N)}(t)$. Mean values different from zero or one (white or black) signify the spontaneous stochasticity, i.e., a stochastic form of the corresponding variable.}
\label{fig9}
\end{figure}\begin{figure}[htp]
\centering
\includegraphics[width=0.55\textwidth]{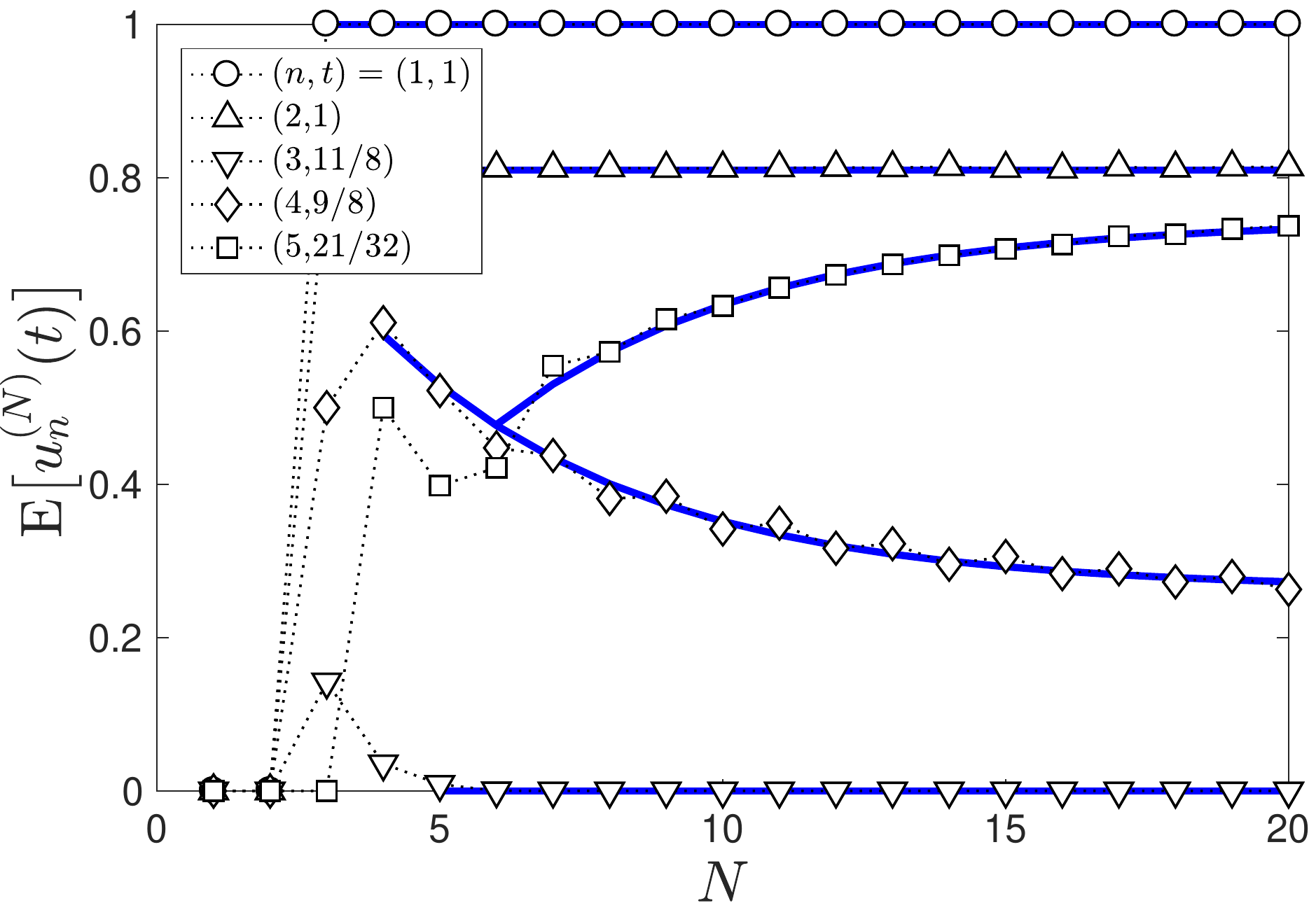}
\caption{Expectations $\mathrm{E}\big[{u}_n^{(N)}(t)\big]$ as functions of the scale number $N$ of the stochastic regularization. The values $(n,t)$ of five selected variables are given in the legend. The results are obtained by averaging with respect to $10^5$ random samples. The blue lines show the exponential convergence to constant values with the same exponent $-0.22$.}
\label{fig10}
\end{figure}

We have no rigorous proof that Theorem~\ref{th3ss} applies in this model. Nevertheless, we performed several numerical tests that verify the  convergence of Theorem~\ref{th3ss} for specific initial and boundary conditions and specific variables ${u}_n(t)$. Figure~\ref{fig9}(a) shows a sample of spontaneously regularized solution in the time interval $0 \le t \le 1.5$, where we took $N = 20$ with the initial conditions ${a} = (0,1,0,0,\ldots)$ and the boundary conditions $(b_0,b_1) = (1,0)$. The second panel of the same plot shows the statistics of such solutions computed by averaging over $10^5$ random samples. Here the color denotes the expectation of the respective variable, which is a probability of having ${u}_n^{(N)}(t) = 1$. The colors other than white or black demonstrate the persistence of the stochastic behavior at large $N$. Existence of the spontaneously stochastic inviscid limit implies that the expectations of ${u}_n^{(N)}(t)$ converge to some values within the interval $[0,1]$ as $N \to \infty$. Next, we verify the rate of convergence with increasing $N$. This is done in Fig.~\ref{fig10} for a selected set of variables, which presents a clear numerical evidence for the convergence to a spontaneously stochastic solution.

As one can infer from Fig.~\ref{fig9}(b), the limiting solution is deterministic at pre-blowup times $t \le T = 0.5$ and becomes spontaneously stochastic at post-blowup times $t > T$.  
This example can be seen as an analogue of spontaneous stochasticity observed in the shell model of turbulence~\cite{mailybaev2016spontaneously,mailybaev2016spontaneous}, where an intrinsically probabilistic solution appears after a finite-time blowup. Numerical tests (not presented here) indicate that the spontaneous stochasticity in our symbolic model is not universal, and the limiting distributions depend on regularization, e.g., on the parameter $p$ of Bernoulli distribution for ${x}_m$. 

\section{Discussion}\label{sec10}

 Motivated by the problem of inviscid limit in fluid dynamics, we considered a class of multi-scale systems with discrete time. These systems model a situation when solutions of the ideal scale-invariant system are non-unique or not globally defined, and the viscous regularization is used for making the system well-posed. We developed the renormalization group (RG) theory for the inviscid limit answering the qualitative questions: why the inviscid limit exists or not, and why it can be universal, i.e., independent of regularization. We associated the inviscid limit with a fixed-point attractor of the RG dynamics, which selects a solution of the ideal system. Universality of this solution follows naturally, because the RG operator (transforming flow maps of regularized systems) depends only on properties of the ideal system, while the form of regularization affects only the initial condition of the RG dynamics. 
 
Our main focus, however, was the situation when the RG attractor is not a fixed point, for example, when the RG dynamics is chaotic. Motivated by microscopic fluctuations, whose importance in turbulence was revealed by Ruelle~\cite{ruelle1979microscopic}, we introduced a small-scale random noise in the viscous range. This yields a stochastic RG operator, which acts on probability (Markov) kernels of stochastically regularized systems. A fixed-point attractor of this RG operator is a probability kernel, which selects a spontaneously stochastic solution of the ideal system in the inviscid limit. This solution is a Markov process solving deterministic equations of the ideal system with deterministic initial and boundary conditions.

Our class of models includes simple symbolic and interacting-phase systems, which we used as solvable examples, but it also refers to realistic systems obtained with a scale-by-scale separation of degrees of freedom~\cite{frisch1999turbulence}. We list below the questions, which we consider important for extending our approach to physical models, e.g., to the Eulerian spontaneous stochasticity of velocity fields observed in the theory of turbulence~\cite{mailybaev2016spontaneously,mailybaev2017toward,biferale2018rayleigh,thalabard2020butterfly}. 
\begin{itemize}
\item[]
\textit{Continuous time:}
how to reformulate the RG theory in terms of infinite-dimensional differential equations with continuous time?
\item[]
\textit{Scale invariance:} we assumed a single space-time scaling symmetry (\ref{eq5}), while the Euler equations of ideal fluid have scaling symmetries (\ref{eq_inv}) with arbitrary exponent $h$, and the intermittent hidden scaling symmetry~\cite{mailybaev2020hidden,mailybaev_thalabard2021hidden}. What is the importance of these symmetries for the inviscid limit?

\item[]
\textit{Conservation laws:}
conserved quantities play important role in the theory of turbulence, e.g., for the dissipative anomaly~\cite{eyink2006onsager}. How these quantities 
can be analysed in the RG theory? 

\item[]
\textit{Dynamics of the RG operator:}
we considered here only two types of explicit models. In the case of coupled expanding maps, our results only hold for linear couplings, and we do not have a general theory in the  case of symbolic models. What can be said about dynamics of the RG operator for these and other scale-invariant models? Other types of dynamical systems questions may involve studying quantities coming from thermodynamic formalism, like entropy and phase transitions. 
We considered only fixed-point attractors of the RG operator. Are there systems with other types of RG attractors? In particular, we observed periodic attractors in symbolic models, and they can be relevant for the inviscid limit in other systems~\cite{mailybaev2016spontaneous,drivas2021life}.  
\end{itemize}

To conclude, it is insightful to identify similarities and differences of our approach with existing RG theories. The existence of a functional RG theory for the inviscid limit of the Navier--Stokes turbulence was anticipated in \cite{eyink1994analogies}, and our work suggests a prototype for such a theory. On one hand, our RG approach resembles the functional RG theory in dynamical systems. In particular, our RG operator contains a rescaled composition of two evolution maps, just as in the Feigenbaum--Cvitanovi\'c functional relation~\cite{feigenbaum1983universal}, and the phenomenon of spontaneous stochasticity is related to a fixed-point attractor of RG dynamics. However, the scaling of our RG operator is in the opposite direction: it zooms out instead of zooming in. This property makes our approach similar to the RG theory for critical phenomena~\cite{wilson1983renormalization}, which studies averaged quantities at increasing observation scales. However, our approach acts in the opposite direction here too: instead of eliminating small degrees of freedom, our RG operator redefines the large-scales evolution while keeping small-scale dynamics intact.

\section{Appendix}

\subsection{Strong and weak solutions in symbolic models}\label{sec3}

In this section we study the initial value problem for a symbolic model with $X = \mathbb{Z} = \{0,1\}$. It is given by relation (\ref{eq4}) on the lattice $(n,t) \in \mathcal{L}$ with initial conditions (\ref{eq4c}) and boundary conditions (\ref{eq4b}).
In partial differential equations, the concept of strong solution refers to functions with sufficient regularity. In particular, it implies a decay of dynamical variables at small scales, e.g., a proper decay of Fourier coefficients for differentiable functions. We extend this concept to our symbolic dynamics by requiring that the variables ${u}_n(t)$ vanish for sufficiently large $n$. Specifically, for any given time $t$ on the lattice, we introduce the leading (possibly infinite) scale number $n_{\max}(t)$ as follows: $n_{\max}(t) = 0$ if all variables ${u}_n(t)$ with $(n,t) \in \mathcal{L}$ vanish at this time; otherwise, 
	\begin{equation}
	\label{eq5b}
	n_{\max}(t) = \max\{n: {u}_n(t) = 1,\ (n,t) \in \mathcal{L}\}.
 	\end{equation}

\begin{definition}
\label{def1}
We call ${u}_n(t)$ a strong solution in a (finite or infinite) interval $t \in [0,T)$ if 
\begin{itemize}
\item[($i$)]
${u}_n(t)$ satisfies equations equations (\ref{eq4})--(\ref{eq4b}) of the ideal system; 
\item[($ii$)] 
$\sup_{t \in I} n_{\max}(t) < \infty$ for any closed interval $I \subset [0,T)$.
\end{itemize}
\end{definition}

Let us denote by $t_{\mathrm{bc}}$ the largest (possibly infinite) time such that $b_t = 0$ for $t < t_{\mathrm{bc}}$. This is the time, at which the boundary condition becomes active. Then, we have the following result:

\begin{proposition}[Strong solutions]
\label{th1}
For any initial and boundary conditions with $n_{\max}(0) < \infty$:
\begin{itemize}
\item[(a)]
If $g(1,0) = 0$ then there exists a unique global-in-time strong solution. 
\item[(b)]
If $g(1,0) = 1$ then there exists a unique strong solution for $0 \le t < T$ with 
	\begin{equation}
	\label{eq5bt}
	T = \left\{\begin{array}{ll}
		t_{\mathrm{bc}}+2, & n_{\max}(0) = 0;\\[3pt]
		2\tau_i, & i = n_{\max}(0) \ge 1.
		\end{array}\right. 
 	\end{equation}
When $T < \infty$, the variables with $n \ge n_{\max}(0)$ and $t < T$ are equal to
	\begin{equation}
	\label{eq5bb}
	{u}_n(t) = \left\{\begin{array}{ll}
		0, & t < T-2\tau_n; \\[2pt]
		1, & t = T-2\tau_n; \\[2pt]
		f(1,0), & t = T-\tau_n.
	\end{array}\right.
 	\end{equation}
Since $n_{\max}(t) = n \to \infty$ as $t = T-2\tau_n \nearrow T$, the strong solution blows up at time $T$. 
\end{itemize}
\end{proposition}

\begin{proof}
First, one  can see that the existence of a strong solution implies its uniqueness. Indeed, the property ($ii$) in Definition~\ref{def1} ensures that any variable $u_n(t)$ with $t < T$ is determined by equations (\ref{eq4})--(\ref{eq4b}) in a finite number of iterations. For the global existence in the item (a) it is sufficient to note that relations (\ref{eq4}) with the condition $g(1,0) = 0$ imply that ${u}_n(t) = 0$ for $n > n_{\max}(0)$ at all times. For proving the finite-time existence in the item (b) one can verify that relations (\ref{eq4}) are satisfied for the solution (\ref{eq5bb}) with  $n \ge n_{\max}(0)$ and $t < T$. As a consequence, relations (\ref{eq4}) define iteratively the remaining variables of the strong solution with $n < n_{\max}(0)$ and $t < T$.
\end{proof}


Solutions with blowup from the Proposition~\ref{th1} cannot be extended beyond the time $T$ in a strong sense. For such an extension, we now introduce the notion of weak solutions, which retain only the first condition $(i)$ in Definition~\ref{def1}, thus, allowing for nonzero components at arbitrarily small scales.

\begin{definition}
We call ${u}_n(t)$ a weak solution if it satisfies equations (\ref{eq4})--(\ref{eq4b}) of the ideal system.
\end{definition}

The immediate consequence of Proposition~\ref{th1}(b) is that the strong solution (\ref{eq5bb}) extends as a weak solution to the blowup time $t = T$ by relations (\ref{eq4}).
In fact, we can prove the global-in-time existence of weak solutions. However, we will see later that these solutions may be nonunique.

\begin{proposition}[Weak solutions]
\label{th2}
For any initial and boundary conditions, there exists a global-in-time weak solution.
\end{proposition}

\begin{proof}
Given arbitrary $N \in \mathbb{N}$, let us introduce a regularized solution ${u}_n^{(N)}(t)$, which satisfies relation (\ref{eq4}) with the initial and boundary conditions for $n \le N$, while all components ${u}^{(N)}_n(t) = 0$ for $n > N$. One can see that such solution is uniquely determined for all times $t \ge 0$. Let us consider a discrete product topology in the space of solutions $\left\{{u}_n(t)\right\}_{(n,t) \in \mathcal{L}} \in X^\mathcal{L}$. The sequential compactness property (see~\cite[Theorem 3.10.35]{engelking1989general}) implies that there exists a subsequence $N_1 < N_2 < N_3 < \cdots$ such that the regularized solutions converge. Here, the convergence implies the (pointwize) convergence for every variable ${u}_n^{(N)}(t)$. The limiting solution satisfies  relation (\ref{eq4}) with initial and boundary conditions at all scales and times.
\end{proof}

In the next example we demonstrate the non-uniqueness of weak solutions, which is triggered by a finite-time blowup.

\begin{proposition}
\label{prop2}
Consider the model B with functions (\ref{eq13}).
For arbitrary boundary and initial conditions from Proposition~\ref{th1}(b) with the blowup time $T < \infty$, the weak solution is unique at times $0 \le t \le T$ and non-unique at $t > T$.
\end{proposition}

\begin{proof}
By inspecting relations (\ref{eq4}) with functions (\ref{eq13}) one can see that: if ${u}_n(t) = 1$, then ${u}'_n = 1$ for odd $t/\tau_n$ and at least one of variables ${u}'_n$, ${u}''_n$ or ${u}''_{n-1}$ is nonzero for even $t/\tau_n$. Using this property iteratively, one can infer that all elements ${u}_n(t)$ of the weak solution vanish for $n > n_{\max}(0)$ and $t < T-2\tau_n$: otherwise, it would contradict to the vanishing initial values ${u}_n(0) = 0$ for $n > n_{\max}(0)$ or the vanishing boundary values $b_t = 0$ for $t < t_{\mathrm{bc}}$.
Then, the remaining variables for $n > n_{\max}(0)$ and $T-2\tau_n \le t \le T$ are defined uniquely by relations (\ref{eq4}); see the proof of Proposition~\ref{th1}(b). The remaining variables with $n \le n_{\max}(0)$ are defined by  (\ref{eq4}). Hence, we proved the uniqueness of the weak solution at pre-blowup times $t \le T$, which coincides with the strong solution at $t < T$.

Our proof of nonuniqueness at $t > T$ is similar to the proof of existence in Proposition~\ref{th2}: we consider solutions ${u}_n^{(N)}(t)$ with two different regularizations and show that they yield different weak solutions in a subsequence limit. First, we choose $N > n_{\max}(0)$ and set all variables ${u}_n^{(N)}(t) = 0$ for $n > N$. One can verify using (\ref{eq4}) and (\ref{eq13}) that relations (\ref{eq5bb}) are valid for $n_{\max}(0) \le n \le N$ and we have 
	\begin{equation}
	\label{eq13p}
	{u}_n^{(N)}(T) = {u}_n^{(N)}(T+\tau_n) = 1,\quad {u}_n^{(N)}(T+2\tau_n) = 0 \quad \textrm{for} \quad n_{\max}(0)+1 < n < N.
 	\end{equation}
As in the proof of Theorem~\ref{th2}, one can choose a convergent subsequence $N_1 < N_2 < \cdots$ providing a weak solution with the property (\ref{eq13p}) for all $n > n_{\max}(0)+1$. Now we choose a different regularization by setting all variables ${u}_n^{(N)}(t) = 0$ for $n > N$ except for ${u}_{N+1}^{(N)}(T) = {u}_{N+1}^{(N)}(T+\tau_{N}) = 1$. One can verify using (\ref{eq4}) and (\ref{eq13}) that relations (\ref{eq5bb}) are valid for $n_{\max}(0) \le n \le N$ and we have 
	\begin{equation}
	\label{eq13q}
	{u}_n^{(N)}(T) = {u}_n^{(N)}(T+\tau_n) = {u}_n^{(N)}(T+2\tau_n) = 1 \quad \textrm{for} \quad n_{\max}(0)+1 < n \le N.
 	\end{equation}
Taking a convergent subsequence, we find a weak solution with the property (\ref{eq13q}) for all $n > n_{\max}(0)+1$. Since the values in (\ref{eq13p}) and (\ref{eq13q}) are different, the two limits yield different weak solutions. 
\end{proof}

\subsection{Solutions at fractional times} \label{sec_fract}

Let us describe regularized solutions and their inviscid limit (\ref{eq32}) at non-integer times on the lattice. Separating the integer and fractional parts, we represent any time on the lattice $\mathcal{L}$ uniquely as
	\begin{equation}
	\label{eq36}
	t = i+\tau_{n_1}+\cdots+\tau_{n_c},
 	\end{equation}
where $i \in \mathbb{N}_0$ and $1 < n_1 < \ldots < n_c$ are increasing scale numbers. 
One can see from Fig.~\ref{fig1} that the variables ${u}_n(t)$ are given by the scale numbers starting with $n = n_c$. Hence, it is convenient to define the corresponding state at time $t$ as the sequence 
	\begin{equation}
	\label{eq16b}
	{{u}}(t) = \left({u}_{n_c}(t),{u}_{n_c+1}(t),{u}_{n_c+2}(t),\ldots\right) \in {X^{\mathbb{N}}},
 	\end{equation}
omitting non-existing components. For example, for $t = 2.625 = 2+\tau_1+\tau_3$, one has ${{u}}(t) = \left({u}_3(t),{u}_4(t),{u}_5(t),\ldots\right)$. 

\begin{proposition}\label{prop8}
At non-integer time (\ref{eq36}), the regularized solution satisfies the following iterative relation
	\begin{equation}
	\label{eqA2_2b}
	{{u}}^{(N)}(t) = 
	\left\{\begin{array}{ll}
	\psi^{(N-n_1+1)} \circ \sigma_+^{n_1-1}\left({{u}}^{(N)}(i)\right), & c = 1; \\[5pt]
	\psi^{(N-n_c+1)} \circ \sigma_+^{n_c-n_{c-1}}\left({{u}}^{(N)}(t-\tau_{n_c})\right), & c \ge 2.
	\end{array}
	\right.
 	\end{equation}
For the inviscid-limit solution from Theorem~\ref{th3}, analogous relation reads
	\begin{equation}
	\label{eqA2_inv}
	{{u}}^\infty(t) = 
	\left\{\begin{array}{ll}
	\psi^\infty \circ \sigma_+^{n_1-1}\left({{u}}^\infty(i)\right), & c = 1; \\[5pt]
	\psi^\infty \circ \sigma_+^{n_c-n_{c-1}}\left({{u}}^\infty(t-\tau_{n_c})\right), & c \ge 2.
	\end{array}
	\right.
 	\end{equation}
\end{proposition}

\begin{proof}
At time $\tau_1$, expression (\ref{eq19flowA}) yields 
	\begin{equation}
	\label{eqA2_1y}
	{{u}}^{(N)}(\tau_1) = \psi^{(N)}({a}). 
 	\end{equation}
At time $\tau_2$, we find the regularized solution using the scaling representation (\ref{eq5N}) with $t = \tau_1$, $\tilde{{a}} = \sigma_+({a})$ and $N-1$ as 
	\begin{equation}
	\label{eqA2_1}
	\begin{array}{rl}
	{{u}}^{(N)}(\tau_2) 
	= & ({{u}}_2^{(N)}(\tau_2),{{u}}_3^{(N)}(\tau_2),\ldots)
	= (\tilde{{u}}_1^{(N-1)}(\tau_1),{{u}}_2^{(N-1)}(\tau_1),\ldots)
	\\[5pt]
	= & \tilde{{{u}}}^{(N-1)}(\tau_1) = \psi^{(N-1)} (\tilde{{a}}) = \psi^{(N-1)} \circ \sigma_+({a}).
	\end{array}
 	\end{equation}
Similarly, for a general turn-over time $\tau_n$ one derives 
	\begin{equation}
	\label{eqA2_2}
	{{u}}^{(N)}(\tau_n) = \psi^{(N-n+1)} \circ \sigma_+^{n-1}({a}).
 	\end{equation}
One can see that the same relation (\ref{eqA2_2}) can be used at every step $\tau_n$ in the expansion (\ref{eq36}) after compensating the missing initial components in (\ref{eq16b}) with an extra shift map. This yields relation (\ref{eqA2_2b}). Relations (\ref{eqA2_inv}) follow from (\ref{eqA2_2b}) by taking the inviscid limit $N \to \infty$.
\end{proof}

Relations (\ref{eqA2_2b}) and (\ref{eqA2_inv}) combined with (\ref{eq19flowAtB}) and (\ref{eqKn4}) determine the regularized and inviscid-limit solutions iteratively at arbitrary time on the lattice. For example, for $t = 2+\tau_1+\tau_3  = 2.625$, we find 
	\begin{equation}
	\label{eqFM2}
	{{u}}^{(N)}(t) = 
	\psi^{(N-2)} \circ \sigma_+^2 \circ \psi^{(N)} 
	\circ \phi^{(N)}_{b_1} \circ \phi^{(N)}_{b_0}({a}).
 	\end{equation}
	\begin{equation}
	\label{eqFM2inf}
	{{u}}^\infty(t) = 
	\psi^\infty \circ \sigma_+^2 \circ \psi^\infty 
	\circ \phi^\infty_{b_1} \circ \phi^\infty_{b_0}({a}).
 	\end{equation}

For the stochastically regularized solution $\mathfrak{{u}}^{(N)}(t)$, maps are replaced by respective probability kernels. Then, compositions and sums of maps become with compositions and convolutions of kernels. Hence, as a consequence of Proposition~\ref{prop8}, we have

\begin{proposition}\label{prop8stat}
For the stochastically regularized solutions at non-integer times (\ref{eq36}), the following iterative relations hold
	\begin{equation}
	\label{eqA2_2bst}
	{P}^{(N)}_t = 
	\left\{\begin{array}{ll}
	\Psi^{(N-n_1+1)} \circ \Sigma_+^{n_1-1} \circ {{P}}^{(N)}_i, & c = 1; \\[5pt]
	\Psi^{(N-n_c+1)} \circ \Sigma_+^{n_c-n_{c-1}} \circ {{P}}^{(N)}_{t-\tau_{n_c}}, & c \ge 2.
	\end{array}
	\right.
 	\end{equation}
For the inviscid-limit solution from Theorem~\ref{th3ss}, analogous relation reads
	\begin{equation}
	\label{eqA2_2bstin}
	{P}^{\infty}_t = 
	\left\{\begin{array}{ll}
	\Psi^{\infty} \circ \Sigma_+^{n_1-1} \circ {{P}}^{\infty}_i, & c = 1; \\[5pt]
	\Psi^{\infty} \circ \Sigma_+^{n_c-n_{c-1}} \circ {{P}}^{\infty}_{t-\tau_{n_c}}, & c \ge 2.
	\end{array}
	\right.
 	\end{equation}
\end{proposition}

Proposition~\ref{prop8stat} with expressions (\ref{eqKn2}) and (\ref{eq41}) determine the statistically regularized and inviscid-limit solutions iteratively at all times. 
For example, at $t = 2+\tau_1+\tau_3  = 2.625$, one finds
	\begin{equation}
	\label{eqFM2st}
	{P}^{(N)}_t = \Psi^{(N-2)} \circ \Sigma_+^2 \circ \Psi^{(N)} 
	\circ \Phi^{(N)}_{b_1} \circ \Phi^{(N)}_{b_0},
 	\end{equation}
	\begin{equation}
	\label{eqFM2stI}
	{P}^\infty_t = \Psi^\infty \circ \Sigma_+^2 \circ \Psi^{\infty} 
	\circ \Phi^{\infty}_{b_1} \circ \Phi^{\infty}_{b_0}.
 	\end{equation}

\vspace{2mm}\noindent\textbf{Acknowledgments.} 
We thank Dmytro Bandak for useful discussions. A.A.M. acknowledges visiting support from the Banff International Research Station, the Simons Center for Geometry and Physics, CNPq grant 308721/2021-7, and FAPERJ grant E-26/201.054/2022. 


\bibliographystyle{plain}
\bibliography{refs}

\end{document}